\documentclass[envcountsame,11pt,envcountsect]{llncs}

\usepackage{xspace}
\usepackage{graphicx}
\usepackage{amsfonts}
\usepackage{amsmath}
\usepackage{latexsym}
\usepackage{amssymb}
\usepackage{color}
\usepackage{comment}
\usepackage{thmtools}
\usepackage{thm-restate}
\usepackage[margin=3cm]{geometry}

\spnewtheorem{myconjecture}{Conjecture}[section]{\bf}{\it}

\newcommand{\torso}{\textup{torso}\xspace}

\newcommand{\cA}{\mathcal{A}}

\newcommand{\cS}{\mathcal{S}}

\newcommand{\bA}{\mathbf{A}}

\newcommand{\bS}{\mathbf{S}}

\newcommand{\hhom}{\textsc{Hom(${H}$)}\xspace}
\newcommand{\lhom}{\textsc{L-Hom(${H}$)}\xspace}
\newcommand{\dlhom}{\textsc{DL-Hom(${H}$)}\xspace}
\newcommand{\dlhoma}[1]{\textsc{DL-Hom(${#1}$)}\xspace}
\newcommand{\dlhomcomp}{\textsc{DL-Hom(${H}$)-Compression}\xspace}
\newcommand{\dlhomdisjcomp}{\textsc{DL-Hom(${H}$)-Disjoint-Compression}\xspace}
\newcommand{\dlhombipcomp}{\textsc{DL-Hom}(${H}$)-\textsc{Bipartite-Compression}\xspace}
\newcommand{\dlhombipcomps}{\textsc{BC}($H$)\xspace}
\newcommand{\dlhomfsfcig}{\textsc{DL-Hom}(${H}$)-\textsc{Fixed-Side-Fixed-Component-Isolated-Good}\xspace}
\newcommand{\dlhomfsfc}{\textsc{DL-Hom}(${H}$)-\textsc{Fixed-Side-Fixed-Component}\xspace}
\newcommand{\dlhomfsfcs}{\textsc{FS-FC(${H}$)}\xspace}

\begin{document}

\title{List H-Coloring a Graph by Removing Few Vertices\thanks{Supported by ERC Starting Grant PARAMTIGHT (No. 280152)}}
\author{Rajesh Chitnis\inst{1}\thanks{Supported in part by NSF CAREER award 1053605, NSF grant CCF-1161626, ONR YIP award N000141110662,
    DARPA/AFOSR grant FA9550-12-1-0423, a University of Maryland Research and Scholarship Award (RASA) and
    a Summer International Research Fellowship from University of Maryland.}
    \and L\'aszl\'o Egri\inst{2} \and D\'aniel Marx\inst{2}}

\institute{Department of Computer Science, University of Maryland at College Park, USA,  \email{rchitnis@cs.umd.edu} \and
  Institute for Computer Science and Control , Hungarian Academy of Sciences
(MTA SZTAKI), Budapest, Hungary. \email{\{dmarx@cs.bme.hu, laszlo.egri@mail.mcgill.ca\}} }

\maketitle

\begin{abstract}
  In the deletion version of the list homomorphism problem, we are given graphs $G$ and $H$, a list $L(v)\subseteq
  V(H)$ for each vertex $v\in V(G)$, and an integer $k$. The task is
  to decide whether there exists a set $W \subseteq V(G)$ of size at
  most $k$ such that there is a homomorphism from $G \setminus W$ to
  $H$ respecting the lists. We show that \dlhom, parameterized by $k$ and $|H|$, is fixed-parameter
  tractable for any $(P_6,C_6)$-free bipartite graph $H$; already for
  this restricted class of graphs, the problem generalizes Vertex
  Cover, Odd Cycle Transversal, and Vertex Multiway Cut parameterized by the size of the cutset and the number of terminals. We conjecture that
  \dlhom is fixed-parameter tractable for the class of graphs $H$ for
  which the list homomorphism problem (without deletions) is
  polynomial-time solvable; by a result of Feder et
  al.~\cite{Feder/et_al:99:LHC}, a graph $H$ belongs to this class precisely if
  it is a bipartite graph whose complement is a circular arc graph. We
  show that this conjecture is equivalent to the fixed-parameter
  tractability of a single fairly natural satisfiability problem,
  \emph{Clause Deletion Chain-SAT}.\end{abstract}

\section{Introduction}

Given two graphs $G$ and $H$ (without loops and parallel edges; unless otherwise stated, we consider only such graphs throughout this paper), a {\em homomorphism} $\phi:G\rightarrow H$ is a mapping $\phi:V(G)\rightarrow V(H)$ such that
$\{u,v\}\in E(G)$ implies $\{\phi(u),\phi(v)\}\in E(H)$; the corresponding algorithmic problem {\em Graph Homomorphism} asks
if $G$ has a homomorphism to $H$. It is easy to see that $G$ has a homomorphism into the clique $K_c$ if and only if $G$ is
$c$-colorable; therefore, the algorithmic study of (variants of) Graph Homomorphism generalizes the study of graph
coloring problems (cf.~Hell and Ne{\v{s}}et{\v{r}}il~\cite{hell-book}). Instead of graphs, one can consider homomorphism
problems in the more general context of relational structures. Feder and Vardi~\cite{DBLP:journals/siamcomp/FederV98}
observed that the standard framework for Constraint Satisfaction Problems (CSP) can be formulated as homomorphism problems for
relational structures. Thus variants of Graph Homomorphism form a rich family of problems that are more general than classical
graph coloring, but does not have the full generality of CSPs.

{\em List Coloring} is a generalization of ordinary graph coloring: for each vertex $v$, the input contains a list $L(v)$ of
allowed colors associated to $v$, and the task is to find a coloring where each vertex gets a color from its list.  In a
similar way, \emph{List Homomorphism} is a generalization of Graph Homomorphism: given two undirected graphs $G, H$ and a list
function $L:V(G)\rightarrow 2^{V(H)}$, the task is to decide if there exists a list homomorphism $\phi:G\rightarrow H$, i.e., a
homomorphism $\phi:G\rightarrow H$ such that for every $v\in V(G)$ we have $\phi(v)\in L(v)$. The List Homomorphism problem
was introduced by Feder and Hell~\cite{FH:98:LHR} and has been studied
extensively~\cite{Egri/et_al:2011:Complexity,Feder/et_al:07:LHG,Feder/et_al:99:LHC,FederHH03,Gutin06:mincosthomomorphism,Hell/Rafiey:11:DLH}. It is also referred to as List $H$-Coloring the graph $G$ since in the special case of $H=K_c$ the problem is equivalent to list coloring where every list is a subset of $\{1,\dots,c\}$.

An active line of research on homomorphism problems is to study the complexity of the problem when the target graph is fixed. Let
$H$ be an undirected graph. The Graph Homomorphism and List Homomorphism problems with fixed target $H$ are denoted by \hhom and \lhom, respectively. A classical result of Hell and Ne{\v{s}}et{\v{r}}il \cite{Hell/Nesetril:90:H-Coloring} states that \hhom\ is
polynomial-time solvable if $H$ is bipartite and NP-complete otherwise. For the more general List Homomorphism problem, Feder
et al.~\cite{Feder/et_al:99:LHC} showed that \lhom\ is in P if $H$ is a bipartite graph whose complement is a
circular arc graph, and it is NP-complete otherwise. Egri et al.~\cite{Egri/et_al:2011:Complexity} further refined this
characterization and gave a complete classification of the complexity of \lhom: they showed that the problem is complete for NP, NL, or L, or otherwise the problem is first-order definable.

In this paper, we increase the expressive power of (list) homomorphisms by allowing a bounded number of vertex deletions from
the left-hand side graph $G$. Formally, in the \dlhom problem we are given as input an undirected graph $G$, an integer $k$, a
list function $L:V(G)\rightarrow 2^{V(H)}$ and the task is to decide if there is a {\em deletion set} $W\subseteq V(G)$ such
that $|W|\leq k$ and the graph $G\setminus W$ has a list homomorphism to $H$. Let us note that \dlhom is NP-hard already
when $H$ consists of a single isolated vertex: in this case the problem is equivalent to \textsc{Vertex Cover}, since removing
the set $W$ has to destroy every edge of $G$.

We study the parameterized complexity of \dlhom parameterized by the number of allowed vertex deletions and the size of the
target graph $H$. We show that \dlhom is fixed-parameter tractable (FPT) for a rich class of target graphs $H$. That is, we
show that \dlhom can be solved in time $f(k,|H|) \cdot n^{O(1)}$ if $H$ is a $(P_6,C_6)$-free bipartite graph, where $f$ is a
computable function that depends only of $k$ and $|H|$ (see \cite{downey-fellows,flum-grohe,niedermeier} for more background
on fixed-parameter tractability). This unifies and generalizes the fixed-parameter tractability of certain well-known problems
in the FPT world:
\begin{itemize}
\item \textsc{Vertex Cover} asks for a set of $k$ vertices whose deletion removes every edge. This problem is equivalent
    to \dlhom where $H$ is a single vertex.
\item \textsc{Odd Cycle Transversal} (also known as \textsc{Vertex Bipartization}) asks for a set of at most $k$ vertices
    whose deletion makes the graph bipartite.  This problem can be expressed by \dlhom\ when $H$ consists of a single
    edge.
\item In \textsc{Vertex Multiway Cut} parameterized by the size of the cutset and the number of terminals, a graph $G$ is given with terminals $t_1,\dots,t_d$, and the task is to find a set of at most $k$ vertices whose deletion disconnects $t_i$ and $t_j$ for any $i \neq j$. This problem can be expressed as \dlhom when $H$ is a matching of $d$ edges, in the following way. Let us obtain $G'$ by subdividing each
    edge of $G$ (making it bipartite) and let the list of $t_i$ contain the vertices of the $i$-th edge $e_i$; all the other
    lists contain every vertex of $H$. It is easy to see that the deleted vertices must separate the terminals
    otherwise there is no homomorphism to $H$ and, conversely, if the terminals are separated from each other, then the
    component of $t_i$ has a list homomorphism to $e_i$.
\end{itemize}
Note that all three problems described above are NP-hard but known to be fixed-parameter
tractable~\cite{almost2sat-3,downey-fellows,DBLP:journals/tcs/Marx06,ReedSV04}.

\textbf{Our Results:}  Clearly, if \lhom is NP-complete, then \dlhom is NP-complete already for $k=0$, hence we cannot expect
it to be FPT. Therefore, by the results of Feder et al.~\cite{Feder/et_al:99:LHC}, \emph{we need to consider only the case when $H$ is a
bipartite graph} whose complement is a circular arc graph. We focus first on those graphs $H$ for which the characterization of
Egri et al.~\cite{Egri/et_al:2011:Complexity} showed that \lhom is not only polynomial-time solvable, but actually in logspace:
these are precisely those bipartite graphs that exclude the path $P_6$ on six vertices and the cycle $C_6$ on six vertices
as induced subgraphs. This class of bipartite graphs admits a decomposition using certain operations (see Section~\ref{sec:solving-fsfc} and \cite{Egri/et_al:2011:Complexity}), and to emphasize this decomposition, we also call this class of graphs {\em
skew decomposable graphs}. Note that the class of skew decomposable graphs is a strict subclass of chordal bipartite graphs ($P_6$ is chordal bipartite but not skew decomposable), and bipartite cographs and bipartite trivially perfect graphs are strict subclasses of skew decomposable graphs.

Our first result is that the \dlhom problem is fixed-parameter tractable for this class
of graphs.
\begin{restatable}{theorem}{mainthm}
\label{thm-main}
\dlhom is FPT parameterized by solution size and $|H|$, if $H$ is restricted to be skew decomposable.
\end{restatable}
Observe that the graphs considered in the examples above are all skew decomposable bipartite graphs, hence
Theorem~\ref{thm-main} is an algorithmic meta-theorem unifying the fixed-parameter tractability of \textsc{Vertex Cover}, \textsc{Odd Cycle
Transversal}, and \textsc{Vertex Multiway Cut} parameterized by the size of the cutset and the number of terminals, and various combinations of these problems.

Theorem~\ref{thm-main} shows that, for a particular class of graphs where \lhom is known to be polynomial-time solvable, the
deletion version \dlhom is fixed-parameter tractable. We conjecture that this holds in general: whenever \lhom is
polynomial-time solvable (i.e., the cases described by Feder et al.~\cite{Feder/et_al:99:LHC}), the corresponding \dlhom problem is
FPT.
\begin{myconjecture}\label{main_conj}
\label{conj:dichotomy} If $H$ is a fixed graph whose complement is a circular arc graph, then \dlhom is FPT parameterized by solution size.
\end{myconjecture}
It might seem unsubstantiated to conjecture fixed-parameter tractability for every bipartite graph $H$ whose complement is a
circular arc graph, but we show that, in a technical sense, proving Conjecture~\ref{conj:dichotomy} boils down to the
fixed-parameter tractability of a single fairly natural problem. We introduce a variant of maximum $\ell$-satisfiability, where the clauses
of the formula are implication chains\footnote{The notation $x_1\to x_2\to \dots \to x_\ell$ is a shorthand for $(x_1 \rightarrow x_2) \wedge (x_2 \rightarrow x_3) \wedge \cdots \wedge (x_{\ell-1} \rightarrow x_\ell)$.} $x_1\to x_2\to \dots \to x_\ell$ of length at most $\ell$, and the task is to make the
formula satisfiable by removing at most $k$ clauses; we call this problem {\em Clause Deletion $\ell$-Chain-SAT
($\ell$-CDCS)} (see Definition~\ref{def:CS-formula}). We conjecture that for every fixed $\ell$, this problem is FPT parameterized by $k$.
\begin{myconjecture}
\label{conj:cdcs} For every fixed $\ell\ge 1$, Clause Deletion $\ell$-Chain-SAT is FPT parameterized by solution size.
\end{myconjecture}
We show that for every bipartite graph $H$ whose complement is a circular arc graph, the problem \dlhom can be reduced to CDCS
for some $\ell$ depending only on $|H|$. Somewhat more surprisingly, we are also able to show a converse statement: for every
$\ell$, there is a bipartite graph $H_\ell$ whose complement is a circular arc graph such that $\ell$-CDCS can be reduced to
\dlhoma{H_\ell}. That is, the two conjectures are equivalent. Therefore, in order to settle
Conjecture~\ref{conj:dichotomy}, one necessarily needs to understand Conjecture~\ref{conj:cdcs} as well. Since the latter
conjecture considers only a single problem (as opposed to an infinite family of problems parameterized by $|H|$), it is likely
that connections with other satisfiability problems can be exploited, and therefore it seems that Conjecture~\ref{conj:cdcs}
is a more promising target for future work.
\begin{theorem}\label{thml:eqv}
Conjectures~\ref{conj:dichotomy} and \ref{conj:cdcs} are equivalent.
\end{theorem}
Note that one may state Conjectures~\ref{conj:dichotomy} and
\ref{conj:cdcs} in a stronger form by claiming fixed-parameter
tractability with two parameters, considering $|H|$ and $\ell$ also as
a parameter (similarly to the statement of
Theorem~\ref{thm-main}). One can show that the equivalence of
Theorem~\ref{thml:eqv} remains true with this a version of the
conjectures as well. However, stating the conjectures with fixed $H$
and fixed $\ell$ gives somewhat simpler and more concrete problems to
work on.

\textbf{Our Techniques:}  For our fixed-parameter tractability results, we use a combination of several techniques (some of
them classical, some of them very recent) from the toolbox of parameterized complexity. Our first goal is to reduce \dlhom to
the special case where each list contains vertices only from one side of one component of the (bipartite) graph $H$; we call
this special case the ``fixed side, fixed component'' version. We note that the reduction to this special case is non-trivial: as the examples above illustrate, expressing \textsc{Odd Cycle Transversal} seems to require that the lists contain vertices from both sides of $H$, and expressing \textsc{Vertex Multiway Cut} seems to require that the lists contain vertices from more than one component of $H$.

We start our reduction by using the standard technique of iterative compression to obtain an instance where, besides a bounded
number of precolored vertices, the graph is bipartite.

We look for obvious conflicts in this instance. Roughly speaking, if there are two precolored vertices $u$ and $v$ in the same component of $G$ with colors $a$ and $b$, respectively, such that either (i) $a$ and $b$ are in different components of $H$, or (ii) $a$ and $b$ are in the same component of $H$ but the parity of the distance between $u$ and $v$ is different from the parity of the distance between $a$ and $b$, then the deletion set must contain a $u-v$ separator. We use the treewidth reduction technique of Marx et
al.~\cite{Marx/et_el:11:FSS} to find a bounded-treewidth region of the graph that contains all such separators. As we know
that this region contains at least one deleted vertex, every component outside this region can contain at most
$k-1$ deleted vertices. Thus we can recursively solve the problem for each such component, and collect all the information that
is necessary to solve the problem for the remaining bounded-treewidth region. We are able to encode our problem as a Monadic Second Order
(MSO) formula, hence we can apply Courcelle's Theorem \cite{Courcelle90} to solve the problem on the bounded-treewidth region.

Even if the instance has no obvious conflicts as described above, we might still need to delete certain vertices
due to more implicit conflicts. But now we know that for each vertex $v$, there is at most one component $C$ of $H$ and one side of $C$ that is consistent with the precolored vertices appearing in the component of $v$, that is, the precolored vertices force this side of $C$ on the vertex $v$. This seems to be close to our goal of being able
to fix a component $C$ of $H$ and a side of $C$ for each vertex. However, there is a subtle detail here: if the deleted set separates a vertex $v$ from every precolored vertex, then the precolored vertices do not force any restriction on $v$. Therefore, it seems
that at each vertex $v$, we have to be prepared for two possibilities: either $v$ is reachable from the precolored
vertices, or not. Unfortunately, this prevents us from assigning each vertex to one of the sides of a single component. We get around
this problem by invoking the ``randomized shadow removal'' technique introduced by Marx and Razgon~\cite{daniel-multicut-arxiv} (and subsequently
used in \cite{icalp-dsfvs,m2,multicut-dags,DBLP:journals/iandc/LokshtanovM13,icalp=parity-multiway}) to modify the instance in
such a way that we can assume that the deletion set does not separate any vertex from the precolored vertices, hence we can
fix the components and the sides.

Note that the above reductions work for any bipartite graph $H$, and the requirement that $H$ be skew decomposable is used
only at the last step: the structural properties of skew decomposable graphs \cite{Egri/et_al:2011:Complexity} allow us to solve the \emph{fixed side fixed component} version of the problem by a simple application of bounded-depth search.

If $H$ is a bipartite graph whose complement is a circular arc graph (recall that this class strictly contains all skew decomposable graphs), then we show how to formulate the problem as an instance of $\ell$-CDCS (showing that Conjecture~\ref{conj:cdcs}
implies Conjecture~\ref{conj:dichotomy}). Let us emphasize that our reduction to $\ell$-CDCS works only if the lists of the \dlhom instance have the
``fixed side'' property, and therefore our proof for the equivalence of the two conjectures (Theorem~\ref{thml:eqv}) utilizes the reduction machinery described above.

\section{Preliminaries}

Given a graph $G$, let $V(G)$ denote its vertices and
$E(G)$ denote its edges. If $G = (U,V,E)$ is bipartite, we call $U$ and $V$ the \emph{sides} of $H$. Let $G$ be a graph and $W \subseteq V(G)$. Then $G[W]$ denotes the subgraph of $G$ induced by the
vertices in $W$. To simplify notation, we often write $G \setminus W$ instead of $G[V(G) \setminus W]$. The set $N(W)$ denotes the
neighborhood of $W$ in $G$, that is, the vertices of $G$ which are not in $W$, but have a neighbor in $W$. Similarly to
\cite{Marx/et_el:11:FSS}, we define two notions of separation: between two sets of vertices and between a pair $(s,t)$ of vertices; note that in the latter case we assume that the separator is disjoint from $s$ and $t$.

\begin{definition}
A set $S$ of vertices {\em separates} the sets of vertices $A$ and $B$ if no component of $G \setminus S$ contains vertices from both $A
\setminus S$ and $B \setminus S$. If $s$ and $t$ are two distinct vertices of $G$, then an $s - t$ separator is a set $S$ of
vertices disjoint from $\{s,t\}$ such that $s$ and $t$ are in different components of $G \setminus S$.
\end{definition}

\begin{definition}
\label{defn:list-homomorphism} Let $G, H$ be graphs and $L$ be a \emph{list function} $V(G) \rightarrow 2^{V(H)}$. A
\textit{list homomorphism} $\phi$ from $(G,L)$ to $H$ (or if $L$ is clear from the context, from $G$ to $H$) is a homomorphism $\phi:G\rightarrow H$ such that $\phi(v)\in L(v)$ for
every $v \in V(G)$. In other words, each vertex $v \in V(G)$ has a list $L(v)$ specifying the possible images of $v$. The right-hand side graph $H$ is called the \emph{target} graph.
\end{definition}

When the target graph $H$ is fixed, we have the following problem:

\begin{center}
\noindent\framebox{\begin{minipage}{4.6in}
\textbf{\lhom}\\
\emph{Input }: A graph $G$ and a list function $L: V(G)\rightarrow 2^{V(H)}$.\\
\emph{Question} : Does there exist a list homomorphism from $(G,L)$ to $H$?
\end{minipage}}
\end{center}

The main problem we consider in this paper is the vertex deletion version of the \lhom problem, i.e., we ask if a set of vertices $W$
can be deleted from $G$ such that the remaining graph has a list homomorphism to $H$. Obviously, the list function is
restricted to ${V(G) \setminus W}$, and for ease of notation, we denote this restricted list function $L|_{V(G) \setminus W}$ by $L \setminus W$. We can now ask the following formal question:

\begin{center}
\noindent\framebox{\begin{minipage}{4.6in}
\textbf{\dlhom}\\
\emph{Input }: A graph $G$, a list function $L: V(G)\rightarrow 2^{V(H)}$, and an integer $k$.\\
\emph{Parameters }: $k$ , $|H|$\\
\emph{Question} : Does there exist a set $W \subseteq V(G)$ of size at most $k$ such that there is a list homomorphism from
$(G \setminus W, L\setminus W)$ to $H$?
\end{minipage}}
\end{center}
Notice that if $k=0$, then \dlhom becomes \lhom. The next section is devoted to prove our first result, Theorem~\ref{thm-main}.
\mainthm*
\section{The Algorithm}
\label{sec:alg}

The algorithm proving Theorem~\ref{thm-main} is constructed through a series of reductions which are outlined in
Figure~\ref{reduction_strategy}. Our starting point is the standard technique of iterative compression.

\begin{figure}[h!tb]
  \centering
  \includegraphics{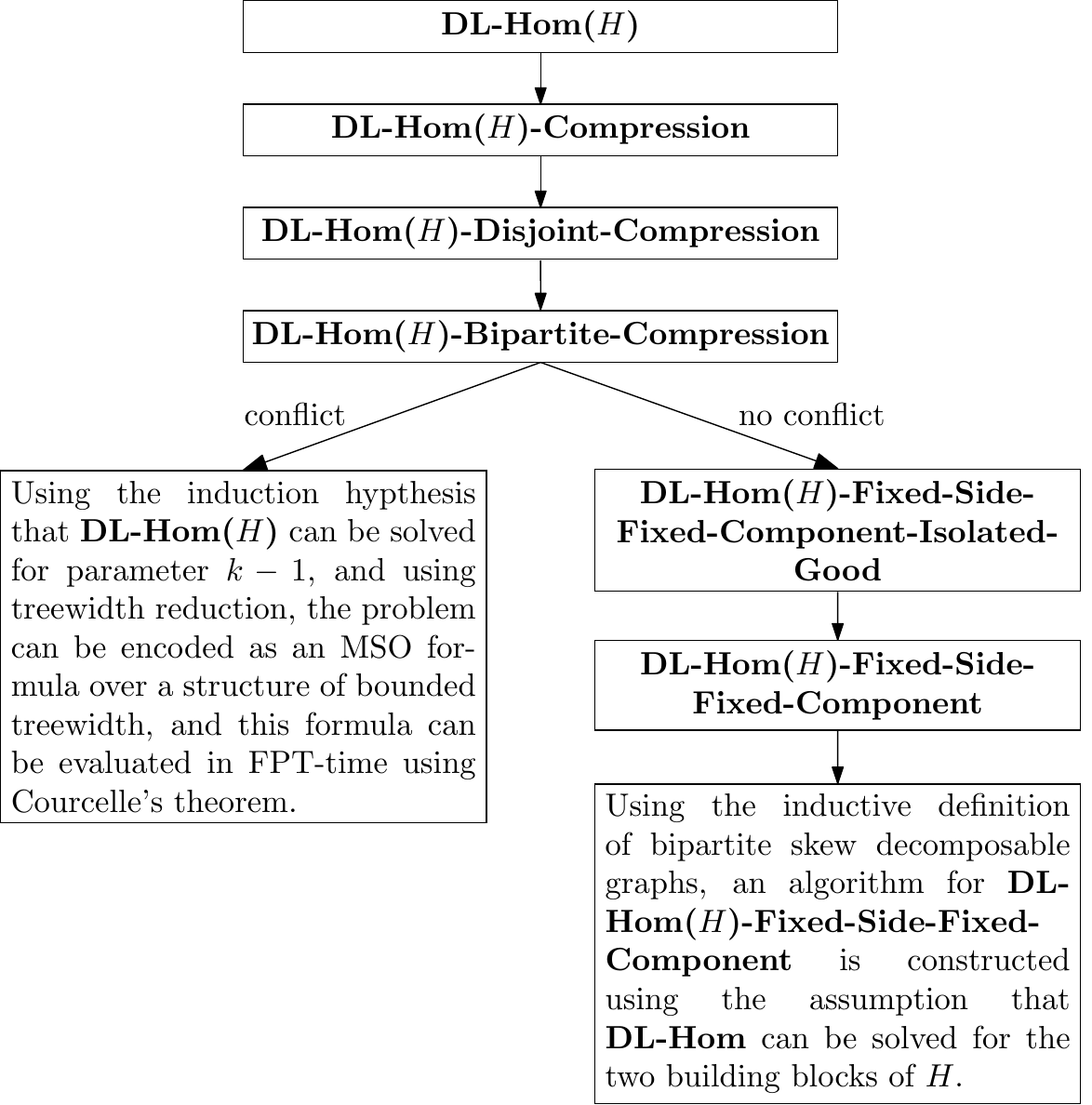}\\
  \caption{The structure of the reductions that establish the fixed-parameter tractability of \dlhom when $H$ is a skew-decomposable graph.}\label{reduction_strategy}
\end{figure}

\subsection{Iterative compression and making the instance bipartite}

We begin with applying the standard technique of \emph{iterative compression} \cite{ReedSV04}, that is, we transform the \dlhom problem into the following problem:

\begin{center}
\noindent\framebox{\begin{minipage}{4.6in}
\textbf{\dlhomcomp}\\
\emph{Input }: A graph $G_0$, a list function $L: V(G_0)\rightarrow 2^{V(H)}$, an integer $k$, and set
$W_0 \subseteq V(G_0)$, $|W_0| \leq k + 1$ such that $(G_0 \setminus W_0, L \setminus W_0)$ has a list homomorphism to $H$.\\
\emph{Parameter }: $k, |H|$\\
\emph{Question} : Does there exist a set $W \subseteq V(G_0)$ with $|W| \leq k$ such that $(G_0 \setminus W, L \setminus W)$ has a list
homomorphism to $H$?
\end{minipage}}
\end{center}

\begin{lemma}
(\textbf{power of iterative compression}) \dlhom can be solved by $O(n)$ calls to an algorithm for \dlhomcomp, where $n$ is
the number of vertices in the input graph. \label{lem:ic}
\end{lemma}
\begin{proof}
Assume that $V(G)=\{v_1,\ldots,v_n\}$ and for $i \in [n]$, let $V_i = \{v_1, \ldots v_i\}$. We construct a sequence of subsets
$X_1 \subseteq V_1, X_2 \subseteq V_2, \dots,X_n \subseteq V_n$ such that $X_i$ is a solution for the instance
$(G[V_i],L|_{V_i},k)$ of \dlhom. In general, we assume that vertices with empty lists are already removed and $k$ is modified
accordingly. Clearly, $X_1=\emptyset$ is a solution for $(G[V_1],L|_{V_1},k)$. Observe, that if $X_i$ is a solution for
$(G[V_i],L|_{V_i},k)$, then $X_i \cup \{v_{i+1}\}$ is a solution for $(G[V_{i+1}],L|_{V_{i+1}},k+1)$. Therefore, for each $i \in
[n-1]$, we set $W = X_i \cup \{v_{i+1}\}$ and use, as a blackbox, an algorithm for \dlhomcomp to construct a solution
$X_{i+1}$ for the instance $(G[V_{i+1}],L|_{V_{i+1}},k)$. Note that if there is no solution for $(G[V_i],L|_{V_i},k)$ for some
$i\in [n]$, then there is no solution for the whole graph $G$. Moreover, since $V_n = V(G)$, if all the calls to the
compression algorithm are successful, then $X_n$ is a solution for the graph $G$ of size at most $k$.
\qed \end{proof}

Now we modify the definition of \dlhomcomp so that it also requires that the solution set in the output be disjoint from the
solution set in the input, and we observe in Lemma~\ref{add_disjointness} below that this can be done without loss of generality.
\begin{center}
\noindent\framebox{\begin{minipage}{4.6in}
\textbf{\dlhomdisjcomp}\\
\emph{Input }: A graph $G_0$, a list function $L: V(G_0)\rightarrow 2^{V(H)}$, an integer $k$, and a
set $W_0 \subseteq V(G_0)$ of size at most $k+1$ such that $G_0 \setminus W_0$ has a list homomorphism to $H$.\\
\emph{Parameters }: $k$, $|H|$\\
\emph{Question} : Does there exist a set $W \subseteq V(G_0)$ disjoint from $W_0$ such that $|W|\leq k$ and $(G_0 \setminus W, L \setminus W)$ has
a list homomorphism to $H$?
\end{minipage}}
\end{center}

\begin{lemma}\label{add_disjointness}
(\textbf{adding disjointness}) \dlhomcomp can be solved by $O(2^{|W_0|})$ calls to an algorithm for the \dlhomdisjcomp
problem, where $W_0$ is the set given as part of the \dlhomcomp instance. \label{lem:disjoint}
\end{lemma}
\begin{proof}
Given an instance $(G,L,W_0,k)$ of \textsc{DL-Hom($H$)-Compression}, we guess the intersection $I$ of $W_0$ and the set $W$ to be
chosen for deletion in the output. We have at most $2^{|W_0|}$ choices for $I$. Then for each guess for $I$, we solve the
\textsc{DL-Hom($H$)-Disjoint-Compression} problem for the instance $(G \setminus I, k-|I|, W_0 \setminus I)$. It is easy to see
that there is a solution $W$ for the \dlhomcomp instance $(G,L,W_0,k)$ if and only if there is a guess $I$ such that
$W\setminus I$ is returned by an algorithm for \dlhomdisjcomp.
\qed \end{proof}

From Lemmas~\ref{lem:ic}~and~\ref{lem:disjoint} it follows that any FPT algorithm for
\textsc{DL-Hom$(H)$-Disjoint-Compression} translates into an FPT algorithm for \textsc{DL-Hom}$(H)$ with an additional blowup
factor of $O(2^{|W_0|}n)$ in the running time. Therefore, in the rest of the paper we will concentrate on giving an FPT algorithm for the
\textsc{DL-Hom$(H)$-Disjoint-Compression} problem.

\label{sec:making-instance-bipartite} Since the new solution $W$ can be assumed to be disjoint from $W_0$, we must have a partial homomorphism from $(G_0[W_0], L|_{W_0})$ to $H$. We guess all such partial list homomorphisms $\gamma$
from$(G_0[W_0], L_{W_0})$ to $H$, and we hope that we can find a set $W$ such that $\gamma$ can be extended to a total list homomorphism
from $(G_0 \setminus W, L \setminus W)$ to $H$. To guess these partial homomorphisms, we simply enumerate all possible mappings from $W_0$ to $H$ and check whether the given mapping is a list homomorphism from $(G_0[W_0], L|_{W_0})$ to $H$. If not we discard the given
mapping. Observe that we need to consider only $|V(H)|^{|W_0|} \leq |V(H)|^{k+1}$ mappings. Hence, in what follows we can
assume that we are given a partial list homomorphism $\gamma$ from $G_0[W_0]$ to $H$.

We propagate the consequences of $\gamma$ to the lists of the vertices in the neighborhood of $W_0$, as follows. Consider a vertex $v \in W_0$. For each neighbor $u$ of $v$ in $N(W_0)$, we trim $L(u)$ as $L(u) \leftarrow L(u) \cap N(\gamma(v))$. Since $H$ is bipartite, the list of each vertex in $N(W_0)$ is now a subset of one of the sides of a single connected component of $H$. We say that such a list is \emph{fixed side} and \emph{fixed component}. Note that while doing this, some of the lists might become empty. We delete those vertices from the graph, and reduce the parameter accordingly.

Recall that $G_0 \setminus W_0$ has a list homomorphism $\phi$ to the bipartite graph $H$, and therefore $G_0
\setminus W_0$ must be bipartite. We will mostly need only the restriction of the homomorphism $\phi$ to $G_0 \setminus (W_0 \cup N(W_0))$, hence we denote
this restriction by $\phi_0$. To summarize the properties of the problem we have at hand, we define it formally below. Note that we do not need the graph $G_0$ and the set $W_0$ any more, only the graph $G_0 \setminus W_0$, and the neighborhood $N(W_0)$. To simplify notation, we refer to $G_0 \setminus W_0$ and $N(W_0)$ as $G$ and $N_0$, respectively.
\begin{center}
\noindent\framebox{\begin{minipage}{4.6in}
\textbf{\dlhombipcomp (\dlhombipcomps)}\\
\emph{Input }: A bipartite graph $G$, a list function $L: V(G)\rightarrow 2^{V(H)}$, a set $N_0 \subseteq V(G)$, where for each $v \in N_0$, the list $L(v)$ is fixed side and fixed component, a list homomorphism $\phi_0$ from $(G \setminus N_0, L \setminus N_0)$ to $H$, and an integer $k$.\\
\emph{Parameters }: $k$, $|H|$\\
\emph{Question} : Does there exist a set $W \subseteq V(G)$, such that $|W| \leq k$ and $(G \setminus W, L \setminus W)$ has a list homomorphism to $H$?
\end{minipage}}
\end{center}

\subsection{The case when there is a conflict}\label{conflict_section}

We define two types of \emph{conflicts} between the vertices of $N_0$. Recall that the lists of the vertices in $N_0$ in a \dlhombipcomps instance are fixed side fixed component.
\begin{definition} \label{defn:conflict}
Let $(G, L, N_0, \phi_0, k)$ be an instance of \dlhombipcomps. Let $u$ and $v$ be vertices in the same component of $G$.
We say that $u$ and $v$ are in \emph{component conflict} if $L(u)$ and $L(v)$ are subsets of vertices of different components
of $H$. Furthermore, $u$ and $v$ are in \emph{parity conflict} if $u$ and $v$ are not in component conflict, and either $u$ and
$v$ belong to the same side of $G$ but $L(u)$ is a subset of one of the sides of a component of $C$ of $H$ and $L(v)$ is a subset of the other
side of $C$, or $u$ and $v$ belong to different sides of $G$ but $L(u)$ and $L(v)$ are subsets of the same side of a component of $H$.
\end{definition}
In this section, we handle the case when such a conflict exists, and the other case is handled in Section~\ref{no_conflict_section}.

If a conflict exists, its presence allows us to invoke the treewidth reduction technique of Marx et al.~\cite{Marx/et_el:11:FSS} to split the instance into a bounded-treewidth part, and into instances having parameter value strictly less than $k$. After
solving these instances with smaller parameter value recursively, we encode the problem in Monadic Second Order logic, and
apply Courcelle's theorem~\cite{Courcelle90}.

The following lemma easily follows from the definitions.
\begin{lemma}\label{lem:conflict-must-be-removed}
Let $(G, L, N_0, \phi_0,k)$ be an instance of \dlhombipcomps. If $u$ and $v$ are any two vertices in $N_0$ that are in
component or parity conflict, then any solution $W$ must contain a set $S$ that separates the sets $\{u\}$ and $\{v\}$.
\end{lemma}

Before we can prove the main lemma of this section (Lemma~\ref{conflict_main_lemma}), first we need the definitions of tree decomposition and treewidth.

\begin{definition}
A \emph{tree decomposition} of a graph $G$ is a pair $(T,\cal{B})$ in which $T$ is
a tree and $\mathcal{B} = \{B_i \;|\; i \in V(T)\}$ is a family of subsets of $V(G)$ such that
\begin{enumerate}
    \item $\bigcup_{i \in V(T)} B_i = V(G)$;
    \item For each $e\in E(G)$, there exists an $i \in V(T)$ such that
        $e \subseteq B_i$;
    \item For every $v \in V(G)$, the set of nodes $\{i \in I \;|\; v \in B_i \}$ forms a connected subtree of $T$.
\end{enumerate}
The \emph{width} of a tree decomposition is the number $\max\{|B_t|-1 \;\mid\; t \in V(T)\}$. The \emph{treewidth} $tw(G)$ is the minimum of the widths of the tree decompositions of $G$.
\end{definition}
It is well known that the maximum clique size of a graph is at most
its treewidth plus one.

A \emph{vocabulary} $\tau$ is a finite set of relation symbols or
predicates. Every relation symbol $R$ in $\tau$ has an arity
associated to it. A relational structure $\bA$ over a vocabulary
$\tau$ consists of a set $A$, called the domain of $\bA$, and a relation
$R^{\bA} \subseteq A^r$ for each $R \in \tau$, where $r$ is the arity
of $R$.

\begin{definition}\label{def:Gaifman_graph}
  The Gaifman graph of a $\tau$-structure $\bA$ is the graph $G_\bA$
  such that $V(G_\bA) = A$ and $\{a,b\}$ ($a \neq b$) is an edge of
  $G_\bA$ if there exists an $R \in \tau$ and a tuple $(a_1,\dots,a_r)
  \in R^\bA$ such that $a,b \in \{a_1,\dots,a_r\}$, where $r$ is the
  arity of $R$. The treewidth of $\bA$ is defined as the treewidth of
  the Gaifman graph of $\bA$.
\end{definition}

The result we need from \cite{Marx/et_el:11:FSS} states that all the
minimal $s-t$ separators of size at most $k$ in $G$ can be covered by a set
$C$ inducing a bounded-treewidth subgraph of $G$. In fact, a stronger
statement is true: this subgraph has bounded treewidth even if we
introduce additional edges in order to take into account connectivity
outside $C$. This is expressed by the operation of taking the torso:
\begin{definition}
Let $G$ be a graph and $C \subseteq V(G)$. The graph $\torso(G,C)$ has vertex set $C$ and two vertices $a,b \in C$ are
adjacent if $\{a,b\} \in E(G)$ or there is a path in $G$ connecting $a$ and $b$ whose internal vertices are not in $C$.
\end{definition}
Observe that by definition, $G[C]$ is a subgraph of $\torso(G,C)$.
\begin{lemma}[\cite{Marx/et_el:11:FSS}]\label{sep_treewidth}
Let $s$ and $t$ be two vertices of $G$. For some $k \geq 0$, let $C_k$ be the union of all minimal sets of size at most $k$ that are $s-t$ separators. There is a $O(g_1(k) \cdot (|E(G) + V(G)|))$ time algorithm that returns a set $C \supset C_k\cup\{s,t\}$ such that
the treewidth of $\torso(G,C)$ is at most $g_2(k)$, for some functions $g_1$ and $g_2$ of $k$.
\end{lemma}
Lemma~\ref{lem:conflict-must-be-removed} gives us a pair of vertices that must be separated. Lemma~\ref{sep_treewidth} specifies a bounded-treewidth region $C$ of the input graph which must contain at least one vertex of the above separator, that is, we know that at least one vertex must be deleted in this bounded-treewidth region.

Courcelle's Theorem gives an easy way of showing that certain problems
are linear-time solvable on bounded-treewidth graphs: it states that
if a problem can be formulated in MSO, then there is a linear-time
algorithm for it. This theorem also holds for relational structures of bounded-treewidth instead of just graphs, a generalization we need because we introduce new relations to encode the properties of the components of $G\setminus C$.
\begin{theorem}(Courcelle's Theorem, see e.g.\ \cite{flum-grohe})\label{thm:Courcelle}
The following problem is fixed parameter tractable:
\begin{center}
\noindent\framebox{\begin{minipage}{4.6in}
$\mathbf{p^*-tw-\textsc{MC(MSO)}}$\\
\emph{\emph{Input }: A structure $\bA$ and an MSO-sentence $\varphi$;}\\
\emph{\emph{Parameter }: $tw(\bA) + |\varphi|$;}\\
\emph{\emph{Problem} : Decide whether $\bA \models \varphi$}.
\end{minipage}}
\end{center}
Moreover, there is a computable function $f$ and an algorithm that solves it in time $f(k,\ell) \cdot |A| + O(|\bA|)$, where
$k = tw(\bA)$ and $\ell = |\varphi|$.
\end{theorem}

The following lemma formalizes the above ideas.
\begin{lemma}\label{conflict_main_lemma}
Let $\mathcal{A}$ be an algorithm that correctly solves \dlhom for input instances in which the first parameter is at most $k-1$. Suppose that the running time of $\mathcal{A}$ is $f(k-1,H) \cdot x^c$, where $x$ is the size of the input, and $c$ is a sufficiently large constant. Let $I$ be an instance of \dlhombipcomps with parameter $k$ that contains a component or parity conflict. Then $I$ can be solved in time $f(k,H) \cdot x^c$ (where $f$ is defined in the proof).
\end{lemma}
\begin{proof}
Let $I = (G, L, N_0, \phi_0, k)$ be an instance of \dlhombipcomps. Let $v,w \in N_0$ such that $v$
and $w$ are in component or parity conflict. Then by
  Lemma~\ref{lem:conflict-must-be-removed}, the deletion set must
  contain a $v-w$ separator. Using Lemma~\ref{sep_treewidth}, we can
  find a set $C$ with the properties stated in the lemma (and note
  that we will also make use of the functions $g_1$ and $g_2$ in the
  statement of the lemma). Most importantly, $C$ contains at least one
  vertex that must be removed in any solution, so the maximum number
  of vertices that can be removed from any connected component of
  $G[V(G) \setminus C]$ without exceeding the budget $k$ is at most
  $k-1$. Therefore, the outline of our strategy is the following. We
  use $\cA$ to solve the problem for some slightly modified versions
  of the components of $G[V(G) \setminus C]$, and using these
  solutions, we construct an MSO formula that encodes our original
  problem $I$.  Furthermore, the relational structure over which this
  MSO formula must be evaluated has bounded treewidth, and therefore
  the formula can be evaluated in linear time using
  Theorem~\ref{thm:Courcelle}.

Assume without loss of generality that $V(H) = \{1,\dots,h\}$. The MSO formula has the form
\begin{multline*}
\exists K_0, \dots, K_h \bigg[ \varphi_{part}(K_0, \dots, K_h) \wedge \varphi_{C}(K_0, \dots, K_h) \wedge \\
  \bigvee_{i=0}^{k} \left( \varphi_{|K_0| \leq i}(K_0) \wedge
    \varphi_{\bar{C},k-i}(K_0, \dots, K_h)\right)\bigg].
\end{multline*}
The set $K_0$ represents the deletion set that is removed from $G[C]$, and
$K_1,\dots,K_h$ specifies the colors of the vertices in the subgraph
$G[C \setminus K_0]$. The sub-formula $\varphi_{part}(K_0, \dots,
K_h)$ checks if $K_0, \dots, K_h$ is a valid partition of $C$, and
$\varphi_{C}$ checks if $K_1,\dots,K_h$ is an $H$-coloring of $G[C
\setminus K_0]$.  The third subformula checks whether there is an
additional set $L \subseteq V(G) \setminus C$ such that $|K_0| + |L|
\leq k$, and the coloring $K_1,\dots,K_h$ of $G[C\setminus K_0]$ can be extended
to $G[V(G) \setminus (K_0 \cup L)]$.  In this part, the formula
$\varphi_{|K_0| \leq i}(K_0)$ checks if the size of $K_0$ is at most
$i$, and the formula $\varphi_{\bar{C},k-i}(K_0, \dots, K_h)$ checks
if the coloring of $G[C\setminus K_0]$ can be extended with $k-i$
additional deletions. Thus the disjunction is true if the set $L$ with
$|K_0|+|L|\le k$ exists.

In what follows, we describe how to construct these subformulas, and we also construct the relational structure $\bS$ from $G$ over which this formula must be evaluated. To simplify the presentation, we refer to $K_0,\dots,K_h$ as
a coloring, even if the vertices in $K_0$ are not mapped to $V(H)$ but
removed.

\textbf{The formula  $\varphi_{part}$.} To check whether $K_0,\dots,K_h$ is a partition of $V(G)$, we use the formula \[\varphi_{part} \equiv \left (\forall x
\bigvee_{i=0}^h K_i(x) \right ) \wedge \left ( \forall x  \bigwedge_{i \neq j} \neg (K_i(x) \wedge K_j(x)) \right ).\]

\textbf{The formula $\varphi_{C}$.} To check
whether a partition $K_0,\dots,K_h$ is a list homomorphism from $G$ to $H$, we encode the lists as follows. For each $T
\subseteq \{1,\dots,h\}$, we produce a unary relation symbol $U_T$. The unary relation $U_T^\bS$ (note that adding a unary
relation to $\bS$ does not increase its treewidth) contains those vertices of $G$ whose list is $T$. The following formula
checks if $K_0,\dots,K_h$ is a list-homomorphism.
\begin{multline*}
\varphi_{C}(K_0,\dots,K_h) \equiv \\ \forall x,y \; \left ( (\neg K_0(x) \wedge \neg K_0(y) \wedge E(x,y)) \rightarrow \left (\bigvee_{(i,j) \in E(H)}  (K_i(x) \wedge K_j(y)) \right ) \right ) \wedge \\
\bigwedge_{i = 1}^h \left ( \forall x \; (K_i(x) \rightarrow \bigvee_{T \ni i} U_T(x) ) \right ).
\end{multline*}

\textbf{The formula $\varphi_{|K_0|\le j}$.}
To check whether $|K_0|\le j$, we use the formula
\[
\varphi_{|K_0|\le j} \equiv
\neg \exists x_1, \dots, x_{j+1} \left ( \bigwedge_{i=1}^{j+1} K_0(x_i) \wedge \bigwedge_{i \neq i'\, 1 \leq i,i' \leq j+1}(x_i \neq x_{i'}) \right ).
\]

\textbf{The formula $\varphi_{\bar{C},j}$.}
First we construct a set of ``indicator'' predicates. For all $q \in \{1,\dots,g_2(k)+1\}$, for all
$q$-tuples $(c_1,\dots,c_q) \in \{0,1,\dots,h\}^q$, and for all $d \in
\{0,\dots,j\}$, we produce a predicate
$R=R_{(c_1,\dots,c_q),d}$ of arity $q$.  Intuitively, the
meaning of a tuple $(v_1,\dots,v_q)$ being in this relation is that if
the clique $\{v_1,\dots,v_q\}$ has the coloring
$(c_1,\dots,c_q)$ (where $c_i=0$ means that the vertex is
deleted), then this coloring can be extended to the components of
$G \setminus C$ that attach precisely to the clique $\{v_1,\dots,v_q\}$ with $d$ further
deletions. Formally, we place a $q$-tuple $(v_1,\dots,v_q) \in
V(G)^q$ into $R$ using the procedure below. (We argue later how to
do this in FPT time.)

Fix an arbitrary ordering $\prec$ on the vertices of $C$. The purpose of $\prec$ will be to avoid
counting the number of vertices that must be removed from a single
component more than once, as we will see later. Let $D$ be
the union of all components of $G[V(G) \setminus C]$ whose neighborhood
in $C$ is precisely $\{v_1,\dots,v_q\}$, and assume without loss of generality that $v_1 \prec \dots \prec v_q$. We call such a union of components a \emph{common neighborhood component}. For each such $D$,
for each $i \in [q]$,  if $c_i \neq 0$, then
for all neighbors $u$ of $v_i$ in $D$, remove any vertex of $H$ from
the list $L(u)$ which is not a neighbor of $c_i$. Let $L'$ be the
new lists obtained this way. Observe that the coloring
$(c_1,\dots,c_q)$ of the vertices $(v_1,\dots,v_q)$ can be
extended to $(D,L)$ after removing $j$ vertices from $D$ if and only
if $(D,L')$ can be $H$-colored after removing $j$ vertices from
$D$. Now we use algorithm $\cA$ to determine the minimum number $z$ of
such deletions. The tuple $(v_1,\dots,v_q)$ is placed into $R$ if $d \ge z$. Observe that if we did not order $\{v_1,\dots,v_q\}$
according to $\prec$, then $\{v_1,\dots,v_q\}$ would be associated
with more than one indicator relation, which would lead to counting
the vertices needed to be removed from $D$ multiple times.

Let $R_1,\dots,R_m$ be an enumeration of all possible
$R_{(c_1,\dots,c_q),d}$ as defined above. Let $\bS$ be the
relational structure $(C; E(G[C]), R_1,\dots,R_m)$. Observe that if
$(v_1,\dots,v_q)$ is a tuple in one of these relations, then
$\{v_1,\dots,v_q\}$ is a clique in $\torso(G,C)$, since it is the
neighborhood of a component of $G\setminus C$. Thus the Gaifman graph
of $\bS$ is a subgraph of $\torso(G,C)$, which means that
$tw(\bS)\leq g_2(k)$. Moreover, for every component of $G\setminus C$,
as its neighborhood in $C$ is a clique in $\torso(G,C)$, the
neighborhood cannot be larger than $g_2(k)+1$: a graph with treewidth at
most $g_2(k)$ has no clique larger than $g_2(k)+1$.

We express the statement that a coloring of $G[C]$ {\em cannot} be
extended to $G\setminus C$ with at most $j$ deletions by stating that
there is a subset of components of $G\setminus C$ such that the total
number of deletions needed for these components is more than $j$.  We
construct a separate formula for each possible way the required number
of deletions can add up to more than $j$ and for each possible
coloring appearing on the neighborhood of these components. Formally,
we define a formula $\psi$ for every combination of
\begin{itemize}
\item integer $0 \le t\le j$ (number of union of components considered),
\item integers $1 \le q_1,\dots, q_t\le g_2(k)+1$ (sizes of the neighborhoods of components),
\item integers $c^i_1,\dots,c^i_{q_i}$ for every $1 \le i \le t$ (colorings of the neighborhoods), and
\item integers $0 \leq d_1,d_2,\dots,d_t \leq j+1$ such that $\sum_{i=1}^t d_i \ge j+1$ (number of deletions required in the neighborhoods)
\end{itemize}
in the following way:
\begin{multline*}
  \psi (K_0,\dots, K_h) \equiv
  \exists x_{1,1},\dots,x_{1,{q_1}},x_{2,1},\dots,x_{2,{q_2}},\dots,x_{t,1},\dots,x_{t,{q_t}} \\
  \bigwedge_{i=1}^t \left (K_{c^i_1}(x_{i,1}) \wedge \dots \wedge
    K_{c^i_{q_i}}(x_{i,q_i}) \wedge
    R_{(c^i_1,\dots, c^i_{q_i}),d_i}(x_{i,1},\dots,x_{i,q_i})
  \right ).
\end{multline*}
Let $\psi_{1},\dots,\psi_{p}$ be an enumeration of all these
formulas. (Notice that the size and the number of these formulas is
bounded by a function of $k$.) We define \[\varphi_{\bar{C},j}(K_0,\dots,
K_h) \equiv \neg \bigvee_{i=1}^{p} \psi_{i}.\]
We argue now that $\varphi_{\bar{C},j}$ is true if and only if it suffices to remove $j$ additional vertices. It follows from
the definition that given an $H$-coloring
$K_0,\dots,K_h$ of $G[C]$, if $\varphi_{\bar C,j}$ is false, then there is a subset of the components $G\setminus C$ witnessing that at
least $j+1$ vertices must be removed from $G[V(G) \setminus C]$ in order to
extend the coloring $K_0,\dots,K_h$ to $G\setminus C$.

Conversely, assume that more than $j$ vertices must be removed from
$G[V(G) \setminus C]$ in order to extend the coloring
$K_0,\dots,K_h$. Then there are neighborhoods $N_1,\dots,N_t\subseteq
C$ with $t \leq j+1$ such that at least $j+1$ vertices must be removed
from the components of $G[V(G) \setminus C]$ whose neighborhoods are
among $N_1,\dots,N_t$. By definition, this is detected by one of the
$\psi_{i}$ in the disjunction, and therefore $\varphi_{\bar{C},j}$ is false.

\textbf{Running time.}
It remains to analyze the running time of the above procedure. By the comments above and by Theorem~\ref{thm:Courcelle}, we
just need to give an upper bound on the time to construct the relations $R_1,\dots,R_m$. First we need to determine the common neighborhood components. Let $D_1,\dots,D_p$ be the components of $G[V(G) \setminus C]$. Find $N(D_1) \cap C$, and find all other components in the list $D_1,\dots,D_p$ having the same neighborhood in $C$ as $D_1$. This produces the common neighborhood component of $D_1$. To find the next common neighborhood component, find the smallest $j$ such that $N(D_j) \cap C \neq N(D_1) \cap C$, and find all other components among $D_1,\dots,D_p$ that have the same neighborhood in $C$ as $D_j$. This produces the common neighborhood component of $D_j$. We repeat this procedure until all common neighborhood components are determined. Let $E_1,\dots,E_n$ be an enumeration of all the common neighborhood components.

Observe that $V(E_i) \cap V(E_j) = \emptyset$ whenever $i \neq j$, implying $\sum_{i=1}^n |V(E_i)| \leq |V(G)|$. For each $E_i$, for all possible colorings of $N(E_i) \cap C$, all possible ways of removing at most $k$ vertices from $N(E_i) \cap C$ (which is at most $\binom{g_2(k) + 1}{k}$), we determine the lists $L'$ as described above. Then we run $\cA$ on $(E_i, L')$ with parameters $0,1,\dots,k-1$ to determine the smallest number of vertices that must be removed. Assume that $N(E_i) \cap C = \{v_1,\dots,v_q\}$, where $v_1 \prec \cdots \prec v_q$.
Then if $(c_1,\dots,c_q)$ is the tuple that encodes the current vertex coloring and the vertices removed from $N(E_i) \cap C$, and $d$ is the smallest number of vertices that must be removed from $E_i$, then $(v_1,\dots,v_q)$ is placed into the relation $R_{(c_1,\dots,c_q),d}$.

The number of times we run $\cA$ for $E_i$ (for different modifications $L'$ of the lists of the vertices of $E_i$) is
$h(k, H)$ for some $h$ depending only on $k$ and $|H|$, and $|N(E_i) \cap C| \leq g_2(k) + 1$. Recall
that the running time of $\cA$ is $f(k-1, H) \cdot x^c$, where $x$ is the size of the input. Therefore the total time $\cA$ is running is
\begin{align*}
\sum_{i=1}^n h(k, H) \cdot f(k-1, H) \cdot |V(E_i)|^c &\leq h(k, H) \cdot f(k-1, H) \cdot \left( \sum_{i=1}^n |V(E_i)| \right)^c \\
&\leq h(k, H) \cdot f(k-1, H) \cdot |V(G)|^c.
\end{align*}
\qed \end{proof}

\subsection{The case when there is no conflict}\label{no_conflict_section}

In this section, given a generic instance $(G,L,N_0,\phi_0,k)$ of \dlhombipcomps, we consider the case when there are no
conflicts among the vertices of $N_0$ (in the sense of Definition~\ref{defn:conflict}). The goal is to prove that it is
sufficient to solve the problem in the case when all the lists are fixed side fixed component. The formal problem definition is given below followed by the theorem we wish to prove.
\begin{center}
\noindent\framebox{\begin{minipage}{4.6in}
\textbf{\dlhomfsfc}, where $H$ is bipartite\\ (\textbf{FS-FC($H$)})\\
\emph{Input }: A graph $G$, a fixed side fixed component list function $L: V(G)\rightarrow 2^{V(H)}$,  and an integer $k$.\\
\emph{Parameters }: $k$, $|H|$\\
\emph{Question} : Does there exist a set $W\subseteq V(G)$ such that $|W|\leq k$ and $G\setminus W$ has a list homomorphism to $H$?
\end{minipage}}
\end{center}
\begin{theorem}\label{th:toig}
  If the \dlhomfsfcs problem is FPT (where $H$ is bipartite), then
  the \dlhom problem is also FPT.
\end{theorem}

Recall that the lists of the vertices in $N_0$ are fixed side
fixed component and $\phi_0$ is a list homomorphism from $G\setminus N_0 \rightarrow H$.
We process the \dlhombipcomps\ instance in the following way. First, if a
component of $G$ does not contain any vertex of $N_0$, then this
component can be colored using $\phi_0$. Hence such components can be
removed from the instance without changing the problem. Consider a
component $C$ of $G$ and let $v$ be a vertex in $C\cap N_0$. Recall
that $L(v)$ is fixed side fixed component by the definition of
\dlhombipcomps; let $H_v$ be the component of $H$ such that $L(v) \subseteq H_v$ in $H$, and let $(S_v,\bar S_v)$ be the bipartition
of $H_v$ such that $L(v) \subseteq S_v$. For every vertex $u$ in $C$ that is in the same
side of $C$ as $v$, let $L'(u)=L(u)\cap S_v$; for every vertex $u$
that is in the other side of $C$, let $L'(u)=L(u)\cap \bar S_v$. Note
that since the instance does not contain any component or parity conflicts, this
operation on $u$ is the same no matter which vertex $v\in C\cap N_0$
is selected: every vertex in $C\cap N_0$ forces $L(u)$ to the same
side of the same component of $H$. The definition of $L'$ is motivated
by the observation that if $u$ remains connected to $v$ in $G\setminus
W$, then $u$ has to use a color from $L'(u)$: its color has to be in
the same component $H_v$ as the colors in $L(v)$, and whether it uses
colors from $S_v$ or $\bar S_v$ is determined by whether it is on the
same side as $L(v)$ or not.

If the fixed side fixed component instance $(G,L',N_0,\phi_0,k)$ has a
solution, then clearly $(G,L,N_0,\phi_0,k)$ has a solution as
well. Unfortunately, the converse is not true: by moving to the more
restricted set $L'$, we may lose solutions. The problem is that even if a
vertex $u$ is in the same side of the same component of $G$ as some
$v\in N_0$, if $u$ is separated from $v$ in $G\setminus W$, then the
color of $u$ does not have to be in the same side of the same
component of $H$ as $L(v)$; therefore, restricting $L(u)$ to $L'(u)$
is not justified. However, we observe that the vertices of $G$ that are separated from
$N_0$ in $G\setminus W$ do not significantly affect the solution: if
$C$ is a component of $G\setminus W$ disjoint from $N_0$, then
$\phi_0$ can be used to color $C$. Therefore, we redefine the problem
in a way that if a component of $G\setminus W$ is disjoint from $N_0$,
then it is ``good'' in the sense that we do not require a coloring for
these components.

\begin{center}
\noindent\framebox{\begin{minipage}{4.6in}
\textbf{\dlhomfsfcig(FS-FC-IG($H$))}\\
\emph{Input }: A graph $G$, a fixed side fixed component list function $L: V(G)\rightarrow 2^{V(H)}$, a set of vertices $N_0\subseteq V(G)$, and an integer $k$.\\
\emph{Parameter }: $k$, $|H|$\\
\emph{Question} : Does there exist a set $W\subseteq V(G)$ such that $|W|\leq k$ and for every component $C$ of
 $G \setminus W$ with $C\cap N_0\neq \emptyset$, there is a list homomorphism from $(G[C],L|_C)$ to $H$?
\end{minipage}}
\end{center}

If the instance $(G,L,N_0,\phi_0,k)$ of \dlhombipcomps\ has a solution, then the modified FS-FC-IG($H$) instance $(G,L',N_0,k)$ also has a solution: for every component $C$ of $G\setminus W$ intersecting $N_0$,
the vertices in $C\cap N_0$ force every vertex of $C$ to respect the more restricted lists $L'$. Conversely, a solution of
instance $(G,L',N_0,k)$ of FS-FC-IG($H$) can be turned into a solution for instance $(G,L,N_0,\phi_0,k)$ of \dlhombipcomps: for
every component of $G\setminus W$ intersecting $N_0$, the coloring using the lists $L'$ is a valid coloring also for the less
restricted lists $L$ and each component disjoint from $N_0$ can be colored using $\phi_0$. Thus we have established a
reduction from \dlhombipcomps\ to FS-FC-IG($H$). In the rest of this section, we further reduce FS-FC-IG($H$) to FS-FC($H$),
thus completing the proof of Theorem~\ref{th:toig}.
\subsubsection{Reducing FS-FC-IG($H$) to FS-FC($H$)}
\label{sec:fsfcig-to-fsfc}

If we could ensure that the solution $W$ has the property that $G\setminus W$
has no component $C$ disjoint from $N_0$, then FS-FC-IG($H$) and
FS-FC($H$) would be equivalent. Intuitively, we would like to remove somehow every
such component $C$ from the instance to ensure this equivalency. This
seems to be very difficult for at least two reasons: we do not know
the deletion set $W$ (finding it is what the problem is about), hence
we do not know where these components are, and it is not clear how to
argue that removing certain sets of vertices does not change the
problem. Nevertheless, the ``shadow removal'' technique of Marx and
Razgon~\cite{daniel-multicut-arxiv} does precisely this: it allows us
to remove components separated from $N_0$ in the solution.

Let us explain how the shadow removal technique can be invoked in our context. We need the following definitions:

\begin{definition}{\bf (closest)}
\label{defn-thin-set} Let $S\subseteq V(G)$. We say that a set $R\supseteq S$ is an $S$-\emph{closest} set if there is no
$R'\subset R$ with $S\subseteq R'$ and $|N(R)|\geq |N(R')|$.
\end{definition}

\begin{definition}{\bf (reach)}
Let $G$ be a graph and $A, X \subseteq V(G)$. Then $R_{G \setminus X}(A)$ is the set of vertices reachable from a vertex in $A$ in the graph $G \setminus X$.
\end{definition}

The following lemma connects these definitions with our problem: we argue that solving FS-FC-IG($H$) essentially requires finding a closest set. We construct a new graph $G'$ from $G$ by adding a new vertex $s$ to $G$, and all edges of the form $\{s, v\}$, $v \in N_0$. Among all solutions of minimum size for FS-FC-IG($H$), fix $W$ to be a solution such that $R_{G' \setminus W}(\{s\}) = R_{G \setminus W}(N_0) \cup \{s\}$ is as small as possible, and set $R = R_{G' \setminus W}(\{s\})$.

\begin{lemma}
It holds that $W=N(R)$, and $R$ is an $\{s\}$-closest set. \label{lem:closest-set}
\end{lemma}
\begin{proof}
  We note that $s \not \in W$. Clearly, $N(R) \subseteq W$. If $W\neq N(R)$, then let us define
  $W'=N(R)$. Now $G\setminus W$ and $G\setminus W'$ have the same
  components intersecting $N_0$: every vertex of $W\setminus W'$ is in
  a component of $G\setminus W'$ that is disjoint from
  $N_0$. Therefore, \textsc{FS-FC-IG}($H$) has a solution with
  deletion set $W'$, contradicting the minimality of $W$.

  If $R$ is not an $\{s\}$-closest set, then there exists a set $R'$
  such that $\{s\} \subseteq R'\subset R$ and $|N(R')|\le
  |N(R)|=|W|$. Let $W'=N(R')$, we have $|W'|\le |W|\le k$. We now claim that $W'$ can be used as a deletion set for  a solution of
  \textsc{FS-FC-IG}($H$). If we show this, then $R_{G'\setminus
    W'}(\{s\})\subseteq R'\subset R$ contradicts the minimality of $W$.

  For a vertex $x$, let $C_{G}(x)$ denote the vertices of the component
  of $G$ that contains $x$. We now show that if $x\in N_0$, then
  $C_{G\setminus W'}(x)\subseteq C_{G\setminus W}(x)$. This shows that
  $W'$ is also a solution, since we know that $W$ is a solution for
  \textsc{FS-FC-IG}($H$), i.e., each component of $G\setminus W$ which
  intersects $N_0$ has a homomorphism to $H$, and hence so does any
  subgraph. Let $x\in N_0$ and $y\in C_{G\setminus W'}(x)$. Then $x,
  y$ are in the same component of $R'$, and hence also in $R$ as
  $R'\subset R$, i.e., $y\in C_{G\setminus W}(x)$.
\qed \end{proof}

The following theorem is the derandomized version of the shadow
removal technique introduced by Marx and Razgon (see Theorem 3.17
of~\cite{daniel-multicut-arxiv}).

\begin{theorem}
  There is an algorithm $\textsc{DeterministicSets}(G,S,k)$ that,
  given an undirected graph $G$, a set $S \subseteq V(G)$, and an
  integer $k$, produces $t=2^{O(k^3)}\cdot \log |V(G)|$ subsets $Z_1,
  Z_2, \ldots, Z_t$ of $V(G)\setminus S$ such that the following
  holds: For every $S$-closest set $R$ with $|N(R)|\leq k$, there is
  at least one $i\in [t]$ such that
\begin{enumerate}
\item $N(R)\cap Z_i=\emptyset$, and
\item $V(G)\setminus (R\cup N(R))\subseteq Z_i$.
\end{enumerate}
\label{thm:shadow-removal}
The running time of the algorithm is $2^{O(k^3)}\cdot n^{O(1)}$.
\end{theorem}

By Lemma~\ref{lem:closest-set} we know that $R=R_{G'\setminus W}(\{s\})$ is an $\{s\}$-closest set.  Thus we can use
Theorem~\ref{thm:shadow-removal} to construct the sets $Z_1$, $\dots$, $Z_t$. Then we branch on choosing one such $Z=Z_i$ and
we can assume in the following that we have a set $Z$ satisfying the following properties:
\begin{equation}
W\cap Z = \emptyset \ \text{ and }\ V(G)\setminus (R\cup W)\subseteq Z.
\tag{$\ast$}
\end{equation}
(Note that $W=N(R)$ implies $V(G)\setminus (R\cup
N(R))=V(G)\setminus (R\cup W)$). That is, $Z$ does not contain any
vertex of the deletion set $W$, but it completely covers the set of
vertices separated from $N_0$ by $W$, and possibly covers some other
vertices not separated from $N_0$. Now we show how to use this
property of the set $Z$ to reduce FS-FC-IG($H$) to FS-FC($H$).

For each component $C$ of $G[Z]$, we run the decision algorithm (see
for example \cite{Feder/et_al:99:LHC}) for \lhom with the list
function $L|_{C}$. If $C$ has no list homomorphism to $H$, then we
call $C$ a \emph{bad component} of $Z$; otherwise, we call $C$ a
\emph{good component} of $Z$. The following lemma shows that all
neighbors of a bad component $C$ in the graph $G \setminus Z$ must be
in the solution $W$.

\begin{lemma}\label{lem:deleting-all-nbrs-is-safe}
  Let $Z$ be a set satisfying $(\ast)$. If $C$ is a bad component of
  $G[Z]$ (i.e., $(C, L|_C)$ has no list homomorphism to $H$), then all
  vertices of the neighborhood of $C$ in $G \setminus Z$ belong to
  $W$.
\end{lemma}
\begin{proof}
  Recall that by assumption, $Z$ contains any vertex that is separated
  from $N_0$ by $W$. Therefore, if a neighbor $v$ of $C$ is in $G
  \setminus Z$, then $v$ is connected to $N_0$ in $G\setminus W$. It
  follows that $C$ is also connected to $N_0$ as $Z$ (and hence $C$)
  is disjoint from $W$. Since $(C, L|_C)$ has no list homomorphism to
  $H$, this contradicts that $W$ is a solution for
  \textsc{FS-FC-IG}($H$).
\qed \end{proof}
By Lemma~\ref{lem:deleting-all-nbrs-is-safe}, we may safely remove the
neighborhood of every bad component $C$ (decreasing the parameter $k$
appropriately) and then, as the component $C$ becomes separated from
$N_0$, we can remove $C$ as well.  We define \[B = \{v \;|\; \text{$v$
  is a vertex in a bad component}\}\] and \[X = \{v \;|\; \text{$v$ is a
  neighbor of a bad component in $G \setminus Z$}\}.\] The following
lemma concludes our reduction.
\begin{lemma}
The instance $(G \setminus (X \cup B), L, k-|X|)$ of FS-FC($H$) is a YES instance if and only if $(G,L,N_0,k)$ is a YES
instance of FS-FC-IG($H$). \label{lem:equivalence-of-2-instances}
\end{lemma}
\begin{proof}
Suppose $(G\setminus (X \cup B), L, k-|X|)$ is a YES instance of FS-FC($H$), and let $W$ be a solution. The set $W \cup
X$ is a solution for the instance $(G,L,N_0,k)$ of FS-FC-IG($H$): every vertex of $B$ is separated from $N_0$ by
$X$, and $W$ is a solution for the rest of $G$. Observe that $|W\cup X| = |W|+|X|\leq (k-|X|)+|X| = k$.

Conversely, suppose that $(G,L,N_0,k)$ is a YES instance of
FS-FC-IG($H$). Choose the same solution $W$ as before, and let
$\varphi$ be a list homomorphism from the components of $G \setminus
W$ that contain a vertex of $N_0$ to $H$. By
Lemma~\ref{lem:deleting-all-nbrs-is-safe}, we have that $X$ is a
subset of $W$.  The size of $W \setminus X$ is clearly at most $k-|X|$.
We claim that $W \setminus X$ is a solution for the instance $(G
\setminus (X \cup B), L, k-|X|)$ of FS-FC($H$), i.e., that there is a
list homomorphism from $\left( G \setminus (X \cup B) \right)
\setminus \left(W \setminus X)\right)= (G \setminus W) \setminus B$ to
$H$.

Recall that $Z$ contains all components separated from $N_0$ by $W$, and for each such component we checked whether there
was a list homomorphism to $H$. The (bad) components which did not have a list homomorphism to $H$ are not present in $(G
\setminus W) \setminus B$. For the (good) components which had a list homomorphism $\psi$ to $H$, we can just obviously use
$\psi$. Since the rest of the components have a vertex from $N_0$, for these components we can use $\varphi$.
\qed \end{proof}
\subsection{Solving the FS-FC($H$) problem for skew decomposable graphs}\label{sec:solving-fsfc}

The last step in our chain of reductions relies on an inductive construction of the bipartite target graph $H$. Recall that we are assuming that neither $P_6$, the path on 6 vertices, nor $C_6$, the cycle on 6 vertices are induced subgraphs of $H$. This is equivalent to assuming that $H$ is skew decomposable, meaning that $H$ admits a certain simple inductive construction (see the definitions below). For any skew decomposable bipartite graph $H$, this construction was used to inductively build a logspace algorithm for \lhom, (Egri et al.~\cite{Egri/et_al:2011:Complexity}). Interestingly, the construction can also be used when we want to obtain an algorithm for FS-FC($H$). We recall the relevant definitions and results from \cite{Egri/et_al:2011:Complexity}. The \emph{special sum} operation is an operation to compose bipartite graphs.

\begin{definition}\label{special_sum}
(\textbf{special sum}) Let $H_1, H_2$ be two bipartite graphs with bipartite classes $T_1, B_1$ and $T_2, B_2$, respectively, such
that neither of $T_1$ or $B_2$ is empty. The \emph{special sum} $H_1 \oslash H_2$ is obtained by taking the disjoint union of
the graphs, and adding all edges $\{u,v\}$ such that $u \in T_1$ and $v \in B_2$.
\end{definition}

\begin{definition}\label{def:sd}
(\textbf{skew decomposable}) A bipartite graph $H$ is called \emph{skew decomposable} if $H \in \cS$, where the graph class
$\cS$ is defined as follows:
  \begin{itemize}
\item $\cS$ contains the graph that is a single vertex;
\item If $H_1, H_2 \in \cS$ then their disjoint union $H_1 \uplus H_2$ also belongs to $\cS$;
\item If $H_1, H_2 \in \cS$ then $H_1 \oslash H_2$ also belongs to $\cS$.
\end{itemize}
\end{definition}

\begin{theorem}[\cite{Egri/et_al:2011:Complexity}]
A bipartite graph $H$ is skew decomposable if and only if neither $P_6$, the path on 6 vertices, nor $C_6$, the cycle on 6 vertices are induced subgraphs of $H$.
\end{theorem}

To give an FPT algorithm for FS-FC($H$), we induct on the construction of $H$ as specified in Definition~\ref{def:sd}. Our induction hypothesis states that if $H = H_1 \uplus H_2$ or $H = H_1 \oslash H_2$, then we already have an algorithm $\cA_i$ for \dlhoma{H_i} with running time $f(H_i,k) \cdot x^c$ (where $x$ is the size of the input and $c$ is a sufficiently large constant), $i \in \{1,2\}$. In the induction step, we use the algorithms $\cA_1$ and $\cA_2$ to construct an algorithm for FS-FC($H$) with running time $f(H,k) \cdot x^c$.

The base case of the induction, i.e.\ when $H$ is a single vertex is just the vertex cover problem (after removing vertices with empty lists and reducing $k$ accordingly). The induction step is taken care of by the following two lemmas.
\begin{lemma}\label{lem:disjoint_union}
  Assume that $H = H_1 \uplus H_2$. Let $\cA_i$ be an algorithm for the problem \dlhoma{H_i} with running time $f(H_i,k)\cdot x^c$, $i \in \{1,2\}$, where $x$ is the size of the input (and $c$ is a sufficiently large constant). Then there is an algorithm for FS-FC($H$) with running time $f(H,k)\cdot x^c$ (where $f(H,k)$ is defined in the proof).
\end{lemma}
\begin{proof}
Let the components of the input graph $G$ be $C_1,\dots,C_n$. For each $C_i$, $i \in [n]$, there is a $j \in \{1,2\}$ such that every vertex of $C_i$ has a list that is a subset of $V(H_j)$ (recall the definition of the FS-FC($H$) problem). We run the algorithm $\cA_j$ at most $k$ times to determine the smallest number $d(C_i)$ such that $d(C_i)$ vertices must be removed from $G[C_i]$ so that it has a list homomorphism to $H$. If $\sum_{i=1}^n d(C_i) > k$ then we reject. Otherwise we accept. The correctness is trivial.

\textbf{Running time.} Assume without loss of generality that $f(H_1,k) \geq f(H_2,k)$. The running time of the algorithm is at most $\sum_{i=1}^n k \cdot f(H_1,k) \cdot |C_i|^c + |G|^d \leq k \cdot f(H_1,k) \cdot |G|^c + |G|^d$, where $|G|^d$ accounts for the overhead calculations (e.g.\ computing the connected components of $G$ and feeding these components to $\cA_1$ or $\cA_2$). The constant $d$ is independent of $k$, so we can assume that $c \geq d$, and set $f(H,k) = k \cdot f(H_1,k) + 1$.
\qed \end{proof}

\begin{lemma}
  Assume that $H = H_1 \oslash H_2$. Let $\cA_i$ be an algorithm for the problem \dlhoma{H_i} with running time $f(H_i,k)\cdot x^c$, $i \in \{1,2\}$, where $x$ is the size of the input (and $c$ is a sufficiently large constant). Then there is an algorithm for FS-FC($H$) with running time $f(H,k)\cdot x^c$ (where $f(H,k)$ is defined in the proof).
\end{lemma}
\begin{proof}
Assume that the bipartite classes of $H_i$ are $T_i, B_i$, $i \in \{1,2\}$. For any $u \in V(G)$ such that $L(u) \subseteq T_1 \cup T_2$ and $L(u) \cap T_1 \neq \emptyset$, we trim $L(u)$ as $L(u) \leftarrow L(u) \cap T_1$. Similarly, for every $v \in V(G)$ such that $L(v) \subseteq B_1 \cup B_2$ and $L(v) \cap B_2 \neq \emptyset$, we trim $L(v)$ as $L(v) \leftarrow L(v) \cap B_2$. Because for any $x_1 \in T_1$ and any $x_2 \in T_2$ it holds that $N(x_1) \supseteq N(x_2)$, and for any $y_1 \in B_1$ and any $y_2 \in B_2$ it holds that $N(y_2) \supseteq N(y_1)$, it is easy to see that reducing the lists this way does not change the solution space.

If $\{u,v\}$ is an edge such that $L(u) \subseteq B_1$ and $L(v) \subseteq T_2$, we call $\{u,v\}$ a \emph{bad edge}. Clearly, we must remove at least one endpoint of a bad edge. We branch on which endpoint of a bad edge to remove until either there are no more bad edges, or we exceed the budget $k$, in which case we abort the current computation branch. Hence, from now on we can assume that there are no bad edges.

Recall that there is a bipartite clique on $T_1$ and $B_2$. This has the consequence that if $\{u,v\}$ is an edge of $G$ such that $L(u) \subseteq T_1$ and $L(v) \subseteq B_2$, then no matter to which element of $L(u)$ the vertex $u$ is mapped, and no matter to which element of $L(v)$ the vertex $v$ is mapped, the edge $\{u,v\}$ is always mapped to an edge of $H$. Therefore we can simply remove these edges from $G$ without changing the solution space. Let $G^\star$ be this modified version of $G$. Observe now that for any connected component $C$ of $G^\star$, for any edge $\{u,v\} \in E(C)$, we have that either $L(u) \subseteq T_1$ and $L(v) \subseteq B_1$, or $L(u) \subseteq T_2$ and $L(v) \subseteq B_2$. That is, no edge can be mapped to the edges between $T_1$ and $B_2$, and therefore without loss of generality, we replace the target graph $H_1 \oslash H_2$ with $H_1 \uplus H_2$. In conclusion, it is sufficient to solve the problem FC-FS($H_1 \uplus H_2$) for $G^\star$ (which we already did in Lemma~\ref{lem:disjoint_union}) and use the obtained solution as a solution for the instance $G$ of FC-FS($H_1 \oslash H_2$) for $G$.

\textbf{Running time.} For every bad edge, we branch on removing one of its endpoints. This needs time $2^k \cdot |G|^{d'}$ for some $d'$, independent of $k$ ($|G|^{d'}$ accounts for finding a bad edge and removing one of its endpoints). If a computation branch ends after removing $k$ vertices and we still have bad edges, that computation branch is terminated. Otherwise we obtain $G^\star$ from $G$ using time $|G|^{d''}$ (this involves trimming the lists, removing edges, etc.), where $d''$ is independent of $k$. To solve FC-FS($H_1 \uplus H_2$) for the instance $G^\star$, we use the same algorithm and the same analysis as in Lemma~\ref{lem:disjoint_union} (with the same notation and assumptions). Noting that $|G^\star| \leq |G|$, this can be done in time $k \cdot f(H_1,k) \cdot |G|^c + |G|^d$. We can assume that $c \geq d,d',d''$, so overall the algorithm runs in
\[2^k \cdot (|G|^{d'} + |G|^{d''} + (k \cdot f(H_1,k) \cdot |G|^c + |G|^d)) \leq 2^k(3 + k \cdot f(H_1,k)) \cdot |G|^c\]
time.
\qed \end{proof}

\section{Relation between DL-Hom(H) and Satisfiability Problems}

The purpose of this section is to prove Theorem~\ref{thml:eqv}: the equivalence of \dlhom with the Clause Deletion $\ell$-Chain SAT ($\ell$-CDCS) problem (defined below), in the cases when \lhom is characterized as polynomial-time solvable by Feder et al.~\cite{Feder/et_al:99:LHC},
that is, when $H$ is a bipartite graph whose complement is a circular arc graph.  This satisfiability problem belongs to the
family of clause deletion problems (e.g., Almost 2-SAT \cite{almost2sat-1,almost2sat-3,almost2sat-2}), where the goal is to
make a formula satisfiable by the deletion of at most $k$ clauses.

\begin{definition}\label{def:CS-formula}
A \emph{chain clause} is a conjunction of the form \[(x_0 \rightarrow x_1) \wedge (x_1 \rightarrow x_2) \wedge \cdots \wedge (x_{m-1} \rightarrow x_m),\]
where $x_i$ and $x_j$ are different variables if $i \neq j$. The \emph{length} of a chain clause is the number of variables it contains. (A chain clause of length $1$ is a variable, and it is satisfied by both possible assignments.)
To simplify notation, we denote chain clauses of the above form as
\[x_0 \rightarrow x_1 \rightarrow \cdots \rightarrow x_m.\]

An $\ell$-Chain-SAT formula consists of:
  \begin{itemize}
  \item a set of variables $V$;
  \item a set of chain clauses over $V$ such that any chain clause has length at most $\ell$;
  \item a set of unary clauses (a unary clause is a variable or its negation).
  \end{itemize}
\end{definition}

\begin{center}
\noindent\framebox{\begin{minipage}{0.98\textwidth}
\textbf{\textsc{Clause Deletion $\ell$-Chain-SAT} \hspace{5mm}($\ell$-CDCS)}\\
\emph{Input }: An $\ell$-Chain-SAT formula $F$.\\
\emph{Parameter }: $k$\\
\emph{Question} : Does there exist a set of clauses of size at most $k$ such that removing these clauses from $F$ makes $F$
satisfiable?
\end{minipage}}
\end{center}

\subsection{The variable-deletion version}

For technical reasons, it will be convenient to work with a variant of the problem where instead of constraints, certain sets of variables are allowed to be removed, a certain disjointness condition is required, and chain clauses of length $2$ behave differently from chain clauses having length $1$ or length at least $3$:
\begin{center}
\noindent\framebox{\begin{minipage}{0.98\textwidth}
\textbf{\textsc{Variable Deletion $\ell$-Chain-SAT}\hspace{5mm}($\ell$-VDCS)}\\
\emph{Input }: An $\ell$-Chain-SAT instance $F$ in which chain clauses of length other than $2$ are on disjoint sets of variables. Furthermore, any variable of $F$ must appear in some chain clause of length different from $2$, and for any chain clause $x \rightarrow y$ of length $2$, the variables $x$ and $y$ cannot both appear in any chain clause of length at least $3$.\\
\emph{Parameter }: $k$\\ \emph{Question} : Does there exist a set of chain clauses of size at most $k$ but not containing any chain clause of length $2$ in $F$ such that removing all variables of these chain clauses, and also removing any clause that contains any of these variables, makes $F$ satisfiable?
\end{minipage}}
\end{center}

The following two lemmas show the equivalence of the two versions of the problem. Note
that the first reduction increases the value of $\ell$, but the
equivalence holds in the sense that $\ell$-CDCS is FPT for every fixed
$\ell$ if and only if $\ell$-VDCS is FPT for every fixed $\ell$.

To simplify the exposition of the proofs that follow, we introduce some terminology. In the context of the CDCS problem, we say that a clause is \emph{undeletable} if it has at least $k+1$ identical copies in the given formula, where $k$ is the maximum number of clauses allowed to be removed. By ``adding an undeletable clause'' we mean adding $k+1$ copies of the given clause, and by ``making a clause undeletable'' we mean adding sufficiently many copies of that clause so that it becomes undeletable. Furthermore, sometimes we refer to chain clauses of length $2$ as \emph{implicational} clauses, and chain clauses of length different from $2$ as \emph{ordinary} clauses.

\begin{lemma}\label{variable->constraint}
There is a parameterized reduction from $\ell$-VDCS to $(2 \ell + 2)$-CDCS.
\end{lemma}
\begin{proof}
Let $F$ be a given $\ell$-VDCS instance. Before we construct a CDCS instance with parameter $k$, we transform $F$ into a more
standard form. First we obtain $F_1$ from $F$ as follows. Let $x_0 \rightarrow \cdots \rightarrow x_m$ be an ordinary clause of $F$ with the property that there exist indices $0 \leq i \leq j \leq m$ such that $F$ contains the unary clauses
$x_i$ and $\neg x_j$. The variables of such a clause must be removed in any solution, so we remove all variables $x_0, \dots, x_m$ from $F$ (and any clauses that contain any of these variables), and add a new chain clause $x'$ of length $1$ (where $x'$ is a new variable), together with unary clauses $x'$ and $\neg x'$. $F_1$ is clearly equivalent to $F$.

We produce a formula $F_2$ from $F_1$ as follows. We mark all the ordinary clauses of $F_1$. Let $C = x_0 \rightarrow \cdots \rightarrow x_m$ be an ordinary clause such that there are indices $j < i$ such that $F_1$ contains unary clauses $x_i$ and $\neg x_j$. Take the largest index $i$ such that $F_1$ contains the unary clause $\neg x_i$, and the smallest index $j$ such that $F_1$ contains the unary clause $x_j$. Observe that in any satisfying variable assignment of this clause, for any $i' < i$, the variable $x_{i'}$ must take on the value $0$. Therefore we collapse all variables $x_{i'}$, $i' < i$, into the variable $x_i$. Using a similar reasoning, we can collapse all variables $x_{j'}$, $j < j'$  into $x_j$. Note that the previous manipulation could convert a ordinary clause into a clause of length $2$, i.e., an implicational clause. The reason we marked the ordinary clauses of $F_1$ is that even in $F_2$, we still wish to treat these marked clauses of length $2$ as ordinary clauses, i.e., we need to remember that these clauses are coming from clauses in $F$ whose variables could be removed from $F$.

Notice that if $F_2$ is an instance of the ``modified VDCS problem'' in which the clause for which the variables are allowed to be removed are the marked clauses, then $F_2$ can be made satisfiable by $k$ deletions if and only if $F_1$ can be made satisfiable by $k$ deletions, where $F_1$ is considered as an ordinary VDCS instance. We are now ready to construct the $(2 \ell + 2)$-CDCS instance $F'$ from $F_2$.

For each marked $x_0 \rightarrow x_1 \rightarrow \cdots \rightarrow x_m$
of $F_2$, we place a chain clause \[x_0 \rightarrow x_0'
\rightarrow \tilde{x}_0 \rightarrow x_1 \rightarrow x_1' \rightarrow x_2 \rightarrow x_2' \rightarrow \cdots \rightarrow x_{m-1} \rightarrow x_{m-1}' \rightarrow \tilde{x}_m \rightarrow x_m \rightarrow x_m'\] into $F'$.

If $\neg x_0$ is a unary clause in $F_2$, then we add the
unary clause $\neg \tilde{x}_0$ and make it undeletable. Similarly, if $x_m$ is a unary clause in
$F_2$, then we add the unary clause $\tilde{x}_m$ to $F'$, and make it undeletable.

Each unmarked chain clause $x_i \rightarrow y_j$ in $F_2$ yields an undeletable implicational clause $\alpha \rightarrow \beta$ in $F'$, where we define $\alpha$ and $\beta$ as follows. Let $y_0 \rightarrow y_1 \rightarrow \cdots \rightarrow y_n$ be the marked clause in $F_2$ that contains $y_j$. Then the corresponding clause in $F'$ is \[y_0 \rightarrow y_0'
\rightarrow \tilde{y}_0 \rightarrow y_1 \rightarrow y_1' \rightarrow y_2 \rightarrow y_2' \rightarrow \cdots \rightarrow y_{n-1} \rightarrow y_{n-1}' \rightarrow \tilde{y}_n \rightarrow y_n \rightarrow y_n'.\]
Then $\alpha = x_i'$ and $\beta = y_j$. It is easy to see that there is a deletion set for the modified VDCS problem instance $F_2$ of size $k$ if and only if there is a deletion set of size $k$ for the $(2\ell + 2)$-CDCS instance $F'$.
\qed \end{proof}

\begin{lemma}\label{constraint->variable}
There is a parameterized reduction from $\ell$-CDCS to $\ell$-VDCS.
\end{lemma}
\begin{proof}
Let $F$ be the $\ell$-CDCS instance. Observe that we can assume without loss of generality that $F$ contains no chain clauses of length $2$: if $x \rightarrow y$ is a chain clause of length $2$, then we introduce a new variable $w$ and replace $x \rightarrow y$ with the chain clause $x \rightarrow y \rightarrow w$. Clearly, this operation does not change the problem.

We produce the desired CDCS-formula $F'$ as follows. For each variable $x$, let $C_1, \dots, C_q$ be all the clauses (both unary and chain) that contain $x$. We make $q$ copies $x_1,\dots,x_q$ of $x$, and replace $x$ in $C_i$ with $x_i$, $i \in [q]$. For any $C_i$ that is a unary clause, we also add chain clause $x_i$ of length $1$.

Removing a chain clause in $F$ corresponds to removing the variables of the corresponding chain clause in $F'$. Removing a unary clause in $F$ also corresponds to removing the variable in the chain clause associated with that unary clause. The converse is equally easy.

Note that if all chain clauses have length at most $2$ then we would have a reduction from $1$-CDCS or $2$-CDCS to $3$-VDCS. We can avoid this blow-up by solving the $1$-CDCS or $2$-CDCS instance directly, which is not too hard.
\qed \end{proof}

\subsection{Reductions}
Bipartite graphs whose complement is a circular arc graph admit a simple representation (see \cite{FederHH03,Circular88}).
\begin{definition}
\label{defn-complement} The class of bipartite graphs whose complement is a circular arc graph corresponds to the class of graphs that can be represented as follows. Let $C$ be a circle, and $N$ and $S$ be two different points on $C$. A northern arc is an arc that contains $N$ but not $S$. A southern arc is an arc that contains $S$ but not $N$. Each vertex $v \in V(H)$ is represented by a northern or a southern arc $A_v$. The pair $\{u,v\}$ is an edge of $H$ if and only if the arcs $A_v$ and $A_u$ do \emph{not} intersect.
\end{definition}

First we reduce $\ell$-VDCS to \dlhom. In fact, we reduce it to the
special case FS-FC($H$) (making the statement somewhat stronger).

\begin{lemma}
\label{reduction_chain_sat_to_LHom} For every $\ell$, there is a bipartite graph $H_\ell$ whose complement is a circular arc graph such that there is a parameterized reduction
from $\ell$-VDCS to FS-FC($H_\ell$).
\end{lemma}
\begin{proof}
Let $F$ be any instance of $\ell$-VDCS. We construct in parallel a graph $H_\ell$ and an instance $G_F$ of \dlhom such that $F$ is
satisfiable after removing the variables of $k$ clauses if and only if $G_F$ maps to $H_\ell$ after removing $k$ vertices. We
will see that the construction of $H_\ell$ is independent of $F$ and depends only on $\ell$.

To define $H_\ell$, first fix a circle with two different points $N$ and $S$. For each ordinary clause $C$ of $F$, we introduce a
vertex $\alpha(C)$ in (the top partition of) $G_F$. There are at most $\ell+1$ satisfying assignments of the clause $C$, and
therefore we introduce $\ell + 1$ arcs $a_0,\dots,a_{\ell}$ in $H_\ell$ to encode all these possibilities. To define these arcs, we
place $\ell + 1$ points $p_0,\dots,p_{\ell - 1}, p_\ell$ on the semicircle from $S$ to $N$ in the clockwise direction
such that $p_0 \neq S$ and $p_\ell \neq N$. Similarly, we place $\ell+1$ points $q_0, q_1, \dots, q_\ell$ on the semicircle
from $N$ to $S$ in the clockwise direction such that $q_0 \neq N$ and $q_\ell \neq S$. The arc $a_i$ goes from $p_i$ to $q_i$
crossing $N$, $i \in \{0,\dots,\ell\}$. See Figure~\ref{chain_to_lhom}. We call these arcs as \emph{value} arcs.

Between any pair of ordinary clauses, there are at most $\ell^2$ possible implicational clauses. Therefore for all possible pairs
$(i,j)$, $0 \leq i,j \leq \ell - 1$ we introduce a set of $6$ arcs in $H_\ell$: $u_{i,j}^1, u_{i,j}^2, v_{i,j}^1, v_{i,j}^2,
w_{i,j}^1, w_{i,j}^2$. The role of these sets of arcs is to simulate the implicational clauses as follows. For each
implicational clause $D$ in $F$, let $C$ and $C'$ be the two (unique) ordinary clauses associated with it. For each such $C$, $C'$ and $D$,
the graph $G_F$ contains a path $ P = \alpha(C) - U(D) - V(D) - W(D) - \alpha(C')$. The clauses $C$ and $C'$ will determine
the lists of $\alpha(C)$ and $\alpha(C')$. The clause $D$ will determine $i$ and $j$, and $i$ and $j$ determine the lists of $U(D), V(D)$ and $W(D)$: $L(U(D)) = \{u_{i,j}^1, u_{i,j}^2\}$, $L(V(D)) = \{v_{i,j}^1, v_{i,j}^2\}$ and $L(W(D)) = \{w_{i,j}^1, w_{i,j}^2\}$.

Suppose that $C = x_0 \rightarrow \cdots \rightarrow x_t$, $C' = y_0 \rightarrow \cdots \rightarrow y_{t'}$, and $D = x_r
\rightarrow y_{r'}$. We set the list of $\alpha(C)$ to be $\{a_0,\dots,a_{t+1}\}$, and the list of $\alpha(C')$ to
$\{a_0,\dots,a_{t'+1}\}$. Finally, we define the lists of $U(D), V(D)$ and $W(D)$. We set the list of $U(D)$ to be
$\{u_{r,r'}^1, u_{r,r'}^2\}$, where $u_{r,r'}^1$ is an arc in $H_\ell$ that starts at $S$, and goes clockwise to include $p_r$ but
not $p_{r+1}$. The arc $u_{r,r'}^2$ starts at $S$ and goes anticlockwise to a point $e$ but it does not include $q_\ell$. Arc
$v_{r,r'}^2$ starts at $N$ and goes clockwise until it includes $e$ (but not $S$). Arc $v_{r,r'}^1$ starts at $N$, it goes
anticlockwise to a point $f$ that is (strictly) between $p_0$ and $S$. The arc $w_{r,r'}^1$ starts at $S$ and goes clockwise
until it includes $f$ but not $p_0$. Finally, $w_{r,r'}^2$ starts at $S$ and goes anticlockwise until includes $q_{r'+1}$ but
not $q_{r'}$.

Assume now that $\alpha(C)$ is mapped to a value arc $a_i$ such that $i \leq r$. Then $a_i$ intersects $u_{r,r'}^1$, so
$U(D)$ must be mapped to $u_{r,r'}^2$. The arc $u_{r,r'}^2$ intersects $v_{r,r'}^2$, so $V(D)$ must be mapped to $v_{r,r'}^1$
which in turn intersects $w_{r,r'}^1$, so $W(D)$ must be mapped to $w_{r,r'}^2$. The arc $w_{r,r'}^2$ intersects any value arc
$a_{i'}$ such that $i' > r'$, so $\alpha(C')$ must be mapped to an arc $a_{i'}$ such that $i' \leq r'$.

On the other hand, if $\alpha(C)$ is mapped to a value arc $a_i$ such that $i > r$, then the above ``chain reaction'' is not
triggered, so $\alpha(C')$ can be mapped to any vertex in its list. More precisely, $U(D)$ can be mapped to $u_{r,r'}^1$,
$V(D)$ can be mapped to $v_{r,r'}^2$, and $W(D)$ can be mapped to $w_{r,r'}^1$, which does not intersect any of the value
arcs, so $\alpha(C')$ can be mapped to anything in its list.

The above analysis suggests the following correspondence between variable assignments of $F$ and homomorphisms from $G_F$ to
$H_\ell$. Mapping $\alpha(C)$ to $a_i$ precisely corresponds to the assignment $x_0 = \dots = x_{i-1} = 0$, $x_i = \dots = x_q = 1$
of the variables of $C$. It is clear that using this correspondence, given a satisfying assignment we can construct a
homomorphism and vice versa. The unary clauses are encoded using the lists of the variables corresponding to ordinary clauses.
For example, if $x_i$ is a unary clause and $x_i$ is among the variables of the ordinary clause $C$, then in any valid variable
assignment we must have that $x_i = 1, x_{i+1} = 1, \dots, x_q = 1$. In our interpretation, this corresponds to restricting the
possible images of $\alpha(C)$ to $a_0, a_1, \dots, a_i$, so we simply remove the rest of the arcs form $L(\alpha(C))$.
Similarly, if $\neg x_i$ is a unary clause, then we must remove $a_j$ from $L(\alpha(C))$ for $j \leq i$.

We give a simple example. Assume that $C = x_0 \rightarrow x_1 \rightarrow x_2 \rightarrow x_3 \rightarrow x_4$ and $C' = y_0
\rightarrow y_1 \rightarrow y_2 \rightarrow y_3$ are ordinary clauses and $D = x_2 \rightarrow y_2$ is an implicational clause.
Then given a satisfying variable assignment such that $x_0 \vee x_1 \vee x_2$, e.g.\ $x_0 = 0$ and $x_1 = x_2 = 1$, we assign
$\alpha(C)$ to $a_1$. Because we were given a satisfying assignment, we must have that $y_2 = y_3 = 1$. So for example, if
$y_1$ is the first among $y_0,y_1,y_2$ that has value $1$, then we assign $\alpha(C')$ to $a_1$. We verify that this mapping
can be extended to the other $3$ vertices of the path between $\alpha(C)$ and $\alpha(C')$. Arc $a_1$ intersects $u_{2,2}^1$,
so $U(D)$ can be assigned (only) to $u_{2,2}^2$. Then $V(D)$ can be assigned (only) to $v_{2,2}^1$, and $W(D)$ (only) to
$w_{2,2}^2$, which intersects $a_3$. Therefore $\alpha(C')$ can be assigned to any of $\{a_0,a_1,a_2\}$ (but not to $a_3$ or
$a_4$), corresponding to the three possible ways the variables of $C'$ could be assigned. If $y_1$ is the first having value
$1$, then we assign $\alpha(C')$ to $a_1$.

On the other hand, if $x_0 \vee x_1 \vee x_2$ is false, then we assign $\alpha(C)$ to $a_i$ where $x_i$ is the first variable
in $C$ with value $1$, or if all the variables have value $0$, then $i = 5$. In all these cases, the first variable among
$y_0, y_1, y_2, y_3$ that has value $1$ could be any of these variables, or all variables could be assigned $0$. If the first
variable that has value $1$ is $y_j$, then we assign $\alpha(C')$ to $a_j$. If all variables are $0$, we assign $\alpha(C')$
to $a_4$. We can check that this mapping can be extended to the variables of the path between $\alpha(C)$ and $\alpha(C')$.

For the converse, $\alpha(C)$ and $\alpha(C')$ are prevented by the path between them to be assigned to value arcs that encode
a variable assignment violating the clauses $C$ and $C'$. If $\alpha(C)$ is mapped to $a_i$ where $i \leq 2$ (i.e.\
all those cases where $x_2$ has value $1$), then there is no homomorphism that maps $\alpha(C')$ to $a_3,a_4$ or $a_5$.

There is one more step in to complete the construction of $G_F$ because we want only vertices corresponding to ordinary clauses to be allowed to be removed. That is, we want to allow only vertices of the form $\alpha(C)$ or $\alpha(C')$ to be removed. To achieve this, for every path $\alpha(C) - U(D) - V(D) - W(D) - \alpha(C')$ we make the inner vertices ``undeletable'' as follows. We replace $U(D)$, $V(D)$, and $W(D)$ with $k+1$ copies. The copies inherit the lists. We add an edge between $\alpha(C)$ and any copy of $U(D)$, an edge between any copy of $U(D)$ and any copy of $V(D)$, an edge between any copy of $V(D)$ and any copy of $W(D)$, and an edge between any copy of $W(D)$ and $\alpha(C')$. This obviously works.

Now we check that the parameters are preserved, i.e., that there is a satisfying assignment of the formula $F$ after removing
the variables corresponding to $k$ ordinary clauses if and only if there is a homomorphism from $G_F$ to $H_\ell$ after removing $k$
vertices. Clearly, if after removing $k$
vertices there is a homomorphism from $G_F$ to $H_\ell$, then we can remove the ordinary clauses from the formula and use the
homomorphism to define a satisfying assignment.

Conversely, if there is a satisfying assignment after removing $k$ ordinary clauses, then we remove the corresponding vertices
from $G_F$ and define a homomorphism from the satisfying assignment. To do this, note that if either endpoint of a path $\alpha(C) - U(D) - V(D) - W(D) - \alpha(C')$ is removed, say $\alpha(C')$, then for any assignment of the remaining end vertex (e.g.\ $\alpha(C)$), we can find images for the vertices $U(D)$, $V(D)$ and $W(D)$ such that each edge of $\alpha(C) - U(D) - V(D) - W(D)$ is mapped to an edge of $H_\ell$. Clearly, this argument also works when we work with the copies of the inner vertices instead of the originals.

$G_F$ is obviously ``fixed side''. Since $H_\ell$ has a single component, $G_F$ is also ``fixed component''. (Observing that
$w_{r,r'}^1$ and $u_{r,r'}^2$ are connected to all the value arcs easily gives that $H_\ell$ is connected.)
\begin{figure}[h!tb]
  \centering
   \includegraphics[scale=0.85]{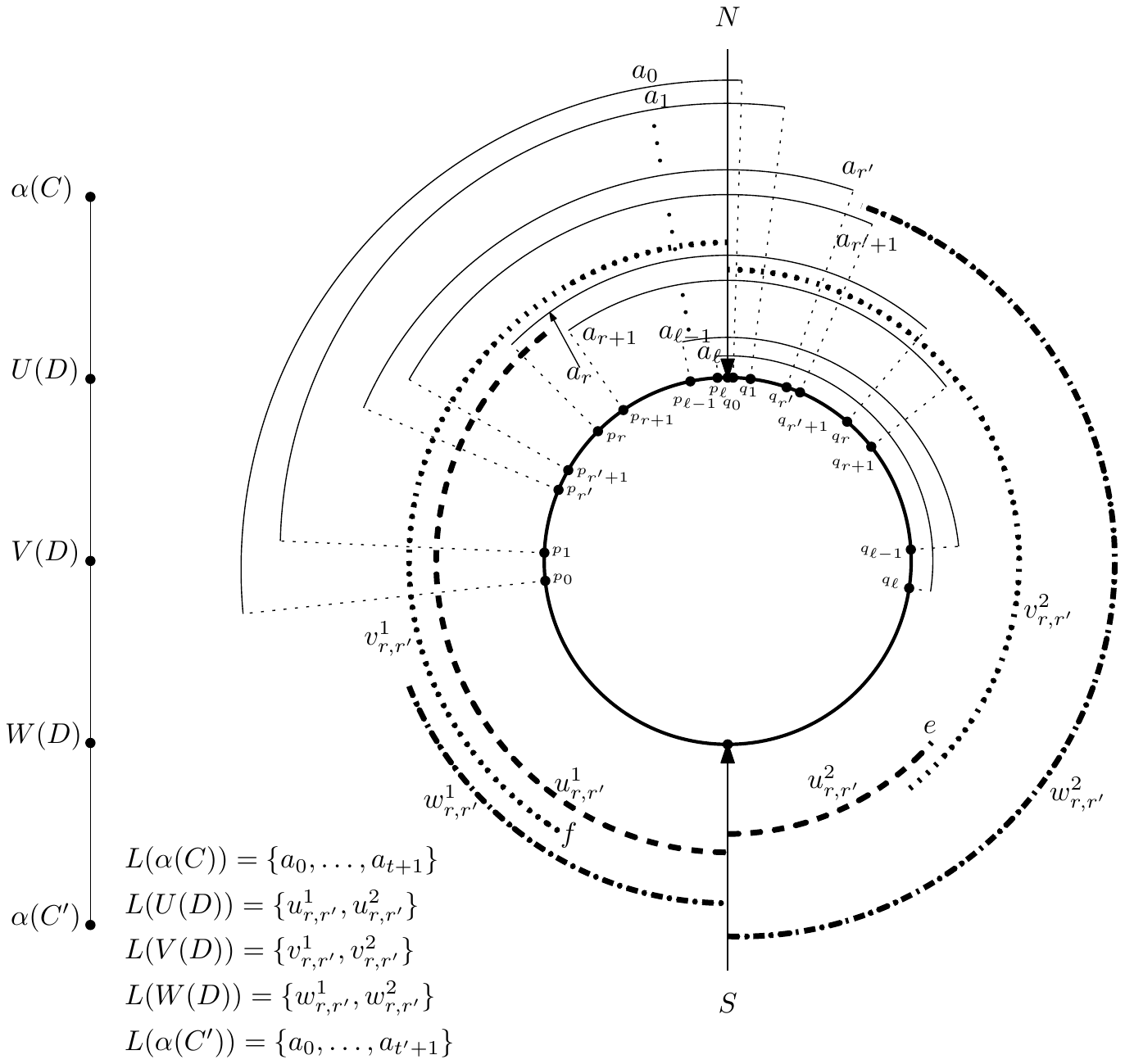}\\
  \caption{Construction of the graph $H_\ell$ and the gadgets in the proof of Lemma~\ref{reduction_chain_sat_to_LHom}.}\label{chain_to_lhom}
\end{figure}
\qed \end{proof}

For the converse direction, the following lemma reduces the special case FS($H$) to $\ell$-VDCS, where FS($H$) is the relaxation of the problem FS-FC($H$) where a list could contain vertices from more than one component of $H$. Note that the chain of
reductions in Section~\ref{sec:alg} works for any bipartite graph $H$ (not only for skew
decomposable bipartite graphs). Thus putting the two reductions together reduced a general instance of \dlhom to $\ell$-VDCS.

\begin{lemma}
Let $H$ be a bipartite graph whose complement is a circular arc graph. Then there is a parameterized reduction from FS($H$) to $\ell$-VDCS, where $\ell = |V(H)| + 1$.
\end{lemma}
\begin{proof}
Let $G$ be an instance of the fixed side problem with bipartition classes $T$ and $B$, and assume we are given a
representation of $H$ as in Definition~\ref{defn-complement}, where the special points on the circle are $N$ and $S$. Let $u
\in T$. Clearly, we can assume that no arc $a \in L(u)$ contains any other arc in $L(u)$. Suppose that $t = |L(u)|$, and that
the arcs in $L(u)$ are $a_0, \dots, a_{t-1}$ (recall that these arcs contain $N$ but not $S$). Furthermore, let $p_i$ and
$q_i$ be points on the circle such that arc $a_i$ is the segment of the circle that begins at $p_i$, goes clockwise passing $N$, and ends at point $q_i$. By renaming the arcs if necessary, we can assume that when we traverse the circle in the
clockwise direction starting at $p_{t-1}$, we visit the endpoints of the arcs in $L(u)$ in the order $p_{t-1},p_{t-2}, \dots,
p_0, q_{t-1},q_{t-2},\dots,q_0$. See Figure~\ref{onion_graph} for an example. Similarly, let $v \in B$ and $t' = |L(v)|$. Let
the arcs in $L(v)$ be $b_0,\dots,b_{t'-1}$, and suppose that $r_i$ and $s_i$ are points on the circle such that the arc
obtained by going from $r_i$ to $s_i$ in the anticlockwise direction (thus traversing $S$) is precisely the arc $b_i$. Just as
before, we assume that traversing the circle in the anticlockwise direction starting at $r_0$, we visit the endpoints of the
arcs in $L(u)$ in the order $r_0,r_1, \dots, r_{t'-1}, s_0,s_1,\dots,s_{t'-1}$.
\begin{figure}[h!tb]
  \centering
  \includegraphics{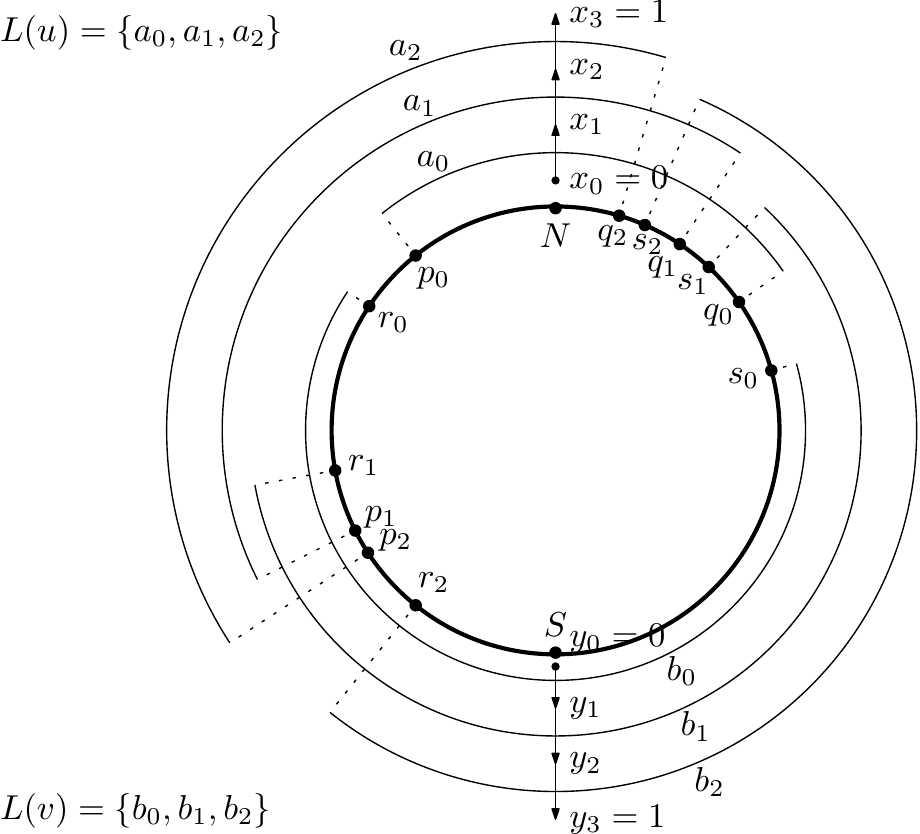}\\
  \caption{The graph $H$ in the proof of Lemma~\ref{onion_graph}.}\label{onion_graph}
\end{figure}

For a vertex $u \in T$, we introduce $t + 1$ ($t = |L(u)|$) variables $x_0,\dots,x_t$ in our VDCS instance. Arc $a_i$ is associated with the pair
$(x_i, x_{i+1})$. We add the chain clause $x_0 \rightarrow \cdots \rightarrow x_t$, and the unary clauses $\neg x_0$, $x_t$.
Notice that if these clauses are satisfied by a variable assignment, then there is a unique index $i$ such that $x_i = 0$ and
$x_{i+1} = 1$. Intuitively, this $0-1$ transition indicates that vertex $u$ of $G$ is assigned to $a_i$. Any vertex $v \in B$
is handled in a similar fashion, i.e.\ we introduce a chain clause $y_0 \rightarrow \dots \rightarrow y_{t'}$ and unary
clauses $\neg y_0$ and $y_{t'}$, where $t' = |L(v)|$. See Figure~\ref{onion_graph}. In the VDCS instance, we will need to allow all these chain clauses to be removed, and at this point, it is possible that some of these clauses have length $2$. To remedy this, for any chain clause $z_0 \rightarrow z_1$ of length $2$ introduced to represent a vertex of $G$, we add a new variable to the clause, i.e.\ $z_0 \rightarrow z_1$ is replaced with $z_0 \rightarrow z_1 \rightarrow w$, for a new variable $w$. Since $z_1$ must always be assigned $1$ in a satisfying assignment, $w$ must also always be assigned value $1$, so the presence of $w$ does not affect in our formula in any way, except that now $z_0 \rightarrow z_1 \rightarrow w$ is allowed to be removed according to the definition of VDCS. (Note that adding a variable to chain clauses of length $2$ to make them length $3$ could result in a $3$-VDCS instance, where $|V(H)|+1 = 2$. In this case, $|V(H)| = 1$, so we can just solve the problem directly.)

We define clauses to encode that edges of $G$ must be mapped to edges of $H$. Let $\{u,v\}$ be an edge of $G$. We interpret
$x_i = 1$ as $u$ being assigned to one of $a_0,\dots,a_{i-1}$, and $x_i = 0$ as $u$ being assigned to one of
$a_i,\dots,a_{t-1}$. Similarly for $v \in B$, $y_j = 1$ is interpreted as assigning $v$ to one of $b_0,\dots,b_{j-1}$, and
$y_j = 0$ as $v$ being assigned to $b_j,\dots,b_{t'-1}$.

We introduce two sets of implicational clauses, the first one consisting of $t$ clauses, and the second one consisting of $t'$
clauses. The role of the first set is to ensure that if $u$ is mapped to an arc $a_i$, then $v$ is mapped to an arc $b_j$ that
does not intersect $a_i$ on the arc from $N$ to $S$ in the clockwise direction. The role of the second set is to ensure that
if $v$ is mapped to an arc $b_j$, then $u$ is mapped to an arc $a_i$ that does not intersect $b_j$ on the arc from $N$ to $S$
in the anticlockwise direction. We define the first set only, as the second set is defined analogously.

The first set of implicational clauses are $x_i \rightarrow y_{j_i}$, one for each $i \in \{1,\dots,t\}$, where $j_i$ is
defined as follows. Let $j'$ be the largest integer such that we can get from $q_{i-1}$ to $s_{j'}$ on the circle in the
clockwise direction without crossing $S$. Then $j_i = j' + 1$. In Figure~\ref{onion_graph}, this corresponds to finding the
``outermost'' arc containing $S$ that ``still'' does not intersect $a_{i-1}$. For example, let $i=2$ in
Figure~\ref{onion_graph}. Then $j_i = 2$, and we obtain the implicational clause $x_2 \rightarrow y_2$. The rest of the
implicational clauses in the figure are $x_1 \rightarrow y_1$ and $x_3 \rightarrow y_3$. In the example in the figure, the
second set of implicational clauses are $y_1 \rightarrow x_1$, $y_2 \rightarrow y_1$, and $y_3 \rightarrow x_3$.

If we have a homomorphism $h$ from $G$ to $H$, we can construct a satisfying assignment for the VDCS formula by setting the
variables as suggested above. That is, if vertex $u$ is in the bipartite class $T$ of $G$ and $u$ is mapped to an arc $a_i(u)
\in L(u)$ then we set $x_j = 0$ for all $j \leq i$, and $x_j = 1$ for all $j > i$. We similarly set the values of the variables of any chain clause associated with a vertex in the bipartite class $B$ of $G$. Notice that this assignment automatically
satisfies all the chain clauses.

Consider an arbitrary implicational clause $x_i \rightarrow y_j$, where $x_i$ is a variable in a chain clause associated with
a vertex $u \in T$, and $y_j$ is a variable associated with a vertex $v \in B$. For the sake of contradiction, assume that $x_i = 1$ and $y_j = 0$. The way we assigned the values of the variables
indicates that $h(u) \in \{a_0,\dots,a_{i-1}\}$, and $v \in \{b_j,\dots,b_{t'}\}$. But because the clause $x_i \rightarrow
y_j$ exists, we know that $\{u, v\}$ is an edge of $G$, so we conclude from the definition of $x_i \rightarrow y_j$ that each
arc $a_0,\dots,a_{i-1}$ is intersected by the arc $b_j$, and also by each of $b_{j+1},\dots,b_{t'}$ because of the ordering of the
arcs. This contradicts the fact that $h$ is a homomorphism.

Conversely, assume that the VDCS formula is satisfied. Then for each vertex of $u \in T$, we find the (unique) index $i$ of the
chain clause associated with $u$ such that $x_i = 0$ and $x_{i+1} = 1$, and we define $h(u)$ to be $a_i$. We similarly define the images of vertices in $B$. To see that $h$ is a
homomorphism, assume for the sake of contradiction that an edge $\{u,v\}$ of $G$ is assigned to a non-edge of $H$. Let $a$ and
$b$ be the arcs to which $u$ and $v$ are assigned, respectively, and let $x_0,\dots,x_t$ and $y_0,\dots,y_{t'}$ be the
variables associated with $u$ and $v$, respectively. Then for some $i$ and $j$, we have that $x_i = 0, x_{i+1} = 1$ and $y_j =
0, y_{j+1}=1$, from our definition of mapping the vertices of $G$ to $V(H)$, we have that $a = a_i$ and $b = b_j$. Assume without loss of generality that $a_i$ and $b_j$ intersect on the semi-circle from $N$ to $S$ going in the clockwise direction. Find the smallest $j'$ such that $b_{j'}$ intersects $a_i$ on the semi-circle from $N$ to $S$ going in the clockwise direction. Then by definition, the implicational clause $x_{i+1} \rightarrow y_{j'}$ is present in the VDCS. Since the formula is satisfied, $y_{j'} = 1$, and because $j' \leq j$, therefore $y_j = 1$, a contradiction.

For the parameters, if there is a homomorphism from $G$ to $H$ after removing $k$ vertices $v_1,\dots,v_k$, then the
$\ell'$-VDCS formula obtained by removing the variables of the chain clauses associated with $v_1,\dots,v_k$ gives exactly
the formula obtained from $G[V(G) \setminus \{v_1,\dots,v_k\}]$ directly. The converse works similarly.
\qed \end{proof}
\section{Concluding Remarks}

The list homomorphism problem is a widely investigated problem in classical complexity theory. In this work, we initiated the study of this problem from the perspective of parameterized complexity: we have shown that the \dlhom is FPT for any skew decomposable graph $H$ parameterized by the solution size and $|H|$, an algorithmic meta-result unifying the fixed parameter tractability of some well-known problems. To achieve this, we welded together a number of classical and recent techniques from the FPT toolbox in a novel way. Our research suggests many open problems, four of which are:
\begin{enumerate}
  \item If $H$ is a fixed bipartite graph whose complement is a circular arc graph, is \dlhom FPT parameterized by solution size? (Conjecture~\ref{main_conj}.)
  \item If $H$ is a fixed digraph such that \lhom is in logspace (such digraphs have been recently characterised in \cite{LNLdicho}), is \dlhom FPT parameterized by solution size?
  \item If $H$ is a matching consisting of $n$ edges, is \dlhom FPT, where the parameter is only the size of the deletion set?
  \item Consider \dlhom for target graphs $H$ in which both vertices with and without loops are allowed. It is known that for such target graphs \lhom is in P if and only if $H$ is a \emph{bi-arc} graph \cite{FederHH03}, or equivalently, if and only if $H$ has a \emph{majority polymorphism}. If $H$ is a fixed bi-arc graph, is there an FPT reduction from \dlhom to $\ell$-CDCS, where $\ell$ depends only on $|H|$?
\end{enumerate}
Note that for the first problem, we already do not know if \dlhom is FPT when $H$ is a path on 7 vertices. (If $H$ is a path on
6 vertices, there is a simple reduction to \textsc{Almost 2-SAT} once we ensure that the instance has fixed side lists.)
Observe that the third problem is a generalization of the \textsc{Vertex Multiway Cut} problem parameterized only by the
cutset. For the fourth problem, we note that the FPT reduction from \dlhom to CDCS for graphs without loops relies
on the fixed side nature of the lists involved. Since the presence of loops in $H$ makes the concept of a fixed side list
meaningless, it is not clear how to achieve such a reduction.

\bibliographystyle{abbrv}
\bibliography{u1}

\begin{thebibliography}{10}

\bibitem{icalp-dsfvs}
R.~H. Chitnis, M.~Cygan, M.~Hajiaghayi, and D.~Marx.
\newblock Directed subset feedback vertex set is fixed-parameter tractable.
\newblock In {\em ICALP (1)}, 2012.

\bibitem{m2}
R.~H. Chitnis, M.~Hajiaghayi, and D.~Marx.
\newblock Fixed-parameter tractability of directed multiway cut parameterized
  by the size of the cutset.
\newblock In {\em SODA}, 2012.

\bibitem{Courcelle90}
B.~Courcelle.
\newblock Graph rewriting: An algebraic and logic approach.
\newblock In J.~van Leeuwen, editor, {\em Handbook of Theoretical Computer
  Science: Volume B: Formal Models and Semantics}, pages 193--242. Elsevier,
  Amsterdam, 1990.

\bibitem{almost2sat-3}
M.~Cygan, M.~Pilipczuk, M.~Pilipczuk, and J.~O. Wojtaszczyk.
\newblock On multiway cut parameterized above lower bounds.
\newblock In {\em IPEC}, pages 1--12, 2011.

\bibitem{downey-fellows}
R.~G. Downey and M.~R. Fellows.
\newblock {\em Parameterized Complexity}.
\newblock Springer-Verlag, 1999.

\bibitem{LNLdicho}
L.~Egri, P.~Hell, B.~Larose, and A.~Rafiey.
\newblock Space complexity of list {H}-colouring: a dichotomy.
\newblock {\em CoRR}, abs/1308.0180, 2013.

\bibitem{Egri/et_al:2011:Complexity}
L.~Egri, A.~A. Krokhin, B.~Larose, and P.~Tesson.
\newblock The complexity of the list homomorphism problem for graphs.
\newblock {\em Theory of Computing Systems}, 51(2):143--178, 2012.

\bibitem{FH:98:LHR}
T.~Feder and P.~Hell.
\newblock List homomorphisms to reflexive graphs.
\newblock {\em J. Comb. Theory, Ser. B}, 72(2):236--250, 1998.

\bibitem{Feder/et_al:99:LHC}
T.~Feder, P.~Hell, and J.~Huang.
\newblock List homomorphisms and circular arc graphs.
\newblock {\em Combinatorica}, 19(4):487--505, 1999.

\bibitem{FederHH03}
T.~Feder, P.~Hell, and J.~Huang.
\newblock Bi-arc graphs and the complexity of list homomorphisms.
\newblock {\em Journal of Graph Theory}, 42(1):61--80, 2003.

\bibitem{Feder/et_al:07:LHG}
T.~Feder, P.~Hell, and J.~Huang.
\newblock List homomorphisms of graphs with bounded degrees.
\newblock {\em Discrete Mathematics}, 307:386--392, 2007.

\bibitem{DBLP:journals/siamcomp/FederV98}
T.~Feder and M.~Y. Vardi.
\newblock The computational structure of monotone monadic {SNP} and constraint
  satisfaction: A study through datalog and group theory.
\newblock {\em SIAM J. Comput.}, 28(1):57--104, 1998.

\bibitem{flum-grohe}
J.~Flum and M.~Grohe.
\newblock {\em Parameterized Complexity Theory}.
\newblock Springer-Verlag, 2006.

\bibitem{Gutin06:mincosthomomorphism}
G.~Gutin, A.~Rafiey, and A.~Yeo.
\newblock Minimum cost and list homomorphisms to semicomplete digraphs.
\newblock {\em Discrete Applied Mathematics}, 154:890--897, 2006.

\bibitem{hell-book}
P.~Hell and J.~Ne{\v{s}}et{\v{r}}il.
\newblock {\em Graphs and homomorphisms}.
\newblock Oxford University Press, 2004.

\bibitem{Hell/Nesetril:90:H-Coloring}
P.~Hell and J.~Ne\v{s}et\v{r}il.
\newblock On the complexity of {$H$}-coloring.
\newblock {\em Journal of Combinatorial Theory, Series B}, 48:92--110, 1990.

\bibitem{Hell/Rafiey:11:DLH}
P.~Hell and A.~Rafiey.
\newblock The dichotomy of list homomorphisms for digraphs.
\newblock In {\em SODA}, pages 1703--1713, 2011.

\bibitem{multicut-dags}
S.~Kratsch, M.~Pilipczuk, M.~Pilipczuk, and M.~Wahlstr{\"o}m.
\newblock Fixed-parameter tractability of multicut in directed acyclic graphs.
\newblock In {\em ICALP (1)}, 2012.

\bibitem{DBLP:journals/iandc/LokshtanovM13}
D.~Lokshtanov and D.~Marx.
\newblock Clustering with local restrictions.
\newblock {\em Inf. Comput.}, 222:278--292, 2013.

\bibitem{almost2sat-2}
D.~Lokshtanov, N.~S. Narayanaswamy, V.~Raman, M.~S. Ramanujan, and S.~Saurabh.
\newblock Faster parameterized algorithms using linear programming.
\newblock {\em CoRR}, abs/1203.0833, 2012.

\bibitem{icalp=parity-multiway}
D.~Lokshtanov and M.~S. Ramanujan.
\newblock Parameterized tractability of multiway cut with parity constraints.
\newblock In {\em ICALP (1)}, pages 750--761, 2012.

\bibitem{DBLP:journals/tcs/Marx06}
D.~Marx.
\newblock Parameterized graph separation problems.
\newblock {\em Theor. Comput. Sci.}, 351(3):394--406, 2006.

\bibitem{Marx/et_el:11:FSS}
D.~Marx, B.~O'Sullivan, and I.~Razgon.
\newblock Finding small separators in linear time via treewidth reduction.
\newblock {\em CoRR}, abs/1110.4765, 2011.

\bibitem{daniel-multicut-arxiv}
D.~Marx and I.~Razgon.
\newblock Fixed-parameter tractability of multicut parameterized by the size of
  the cutset.
\newblock {\em CoRR}, abs/1010.3633, 2010.

\bibitem{niedermeier}
R.~Niedermeier.
\newblock {\em Invitation to Fixed-Parameter Algorithms}.
\newblock Oxford University Press, 2006.

\bibitem{almost2sat-1}
I.~Razgon and B.~O'Sullivan.
\newblock Almost 2-{SAT} is fixed-parameter tractable.
\newblock {\em J. Comput. Syst. Sci.}, 75(8):435--450, 2009.

\bibitem{ReedSV04}
B.~A. Reed, K.~Smith, and A.~Vetta.
\newblock Finding odd cycle transversals.
\newblock {\em Oper. Res. Lett.}, 32(4):299--301, 2004.

\bibitem{Circular88}
J.~Spinrad.
\newblock Circular-arc graphs with clique cover number two.
\newblock {\em J. Comb. Theory, Ser. B}, 44(3):300--306, 1988.

\end{thebibliography}

\end{document}